\documentclass[a4paper,11pt]{amsart}

\usepackage[]{amsmath}
\usepackage{amsfonts}
\usepackage{amssymb}
\usepackage{xcolor}
\usepackage{mathrsfs}

\usepackage{tikz}
\usetikzlibrary{shapes.geometric,positioning,calc,cd}

\usepackage{binarytree} 

\usepackage[margin=3cm]{geometry} 
\usepackage{hyperref}

\usepackage{algorithm}
\usepackage{algpseudocode}

\newtheorem{thm}{Theorem}[section]
\newtheorem{prop}[thm]{Proposition}
\newtheorem{lem}[thm]{Lemma}
\newtheorem{cor}[thm]{Corollary}

\theoremstyle{definition}
\newtheorem{defn}[thm]{Definition}
\newtheorem{rem}[thm]{Remark}
\newtheorem{eg}[thm]{Example}

\newcommand{\RP}{\mathcal{R}\negthinspace\mathcal{P}}
\newcommand{\F}{\mathcal F}
\newcommand{\B}{\mathfrak B}
\newcommand{\Br}{\mathfrak B}
\newcommand{\sS}{\mathfrak S}
\newcommand{\dD}{\mathfrak D}
\newcommand{\LL}{\mathscr L}
\newcommand{\RR}{\mathscr R}
\newcommand{\DD}{\mathscr D}
\newcommand{\HH}{\mathscr H}

\newcommand{\coker}{\mathrm{coker}}
\newcommand{\codom}{\mathrm{codom}}
\newcommand{\dom}{\mathrm{dom}}

\newcommand{\sL}{\LL}
\newcommand{\sH}{\HH}
\newcommand{\sR}{\RR}
\newcommand{\sD}{\DD}

\newcommand{\N}{\mathbb N}
\newcommand{\oT}{\overline{T}}
\newcommand{\ot}{\overline{T}}
\newcommand{\dx}{\alpha}

\title[Brauer diagram models for phylogenetics]{Brauer and partition diagram models for phylogenetic trees and forests}
\author{Andrew Francis}
\address{Centre for Research in Mathematics and Data Science, Western Sydney University}
\email{a.francis@westernsydney.edu.au}
\author{Peter D. Jarvis}
\address{School of Mathematics and Physics, University of Tasmania}
\email{peter.jarvis@utas.edu.au}

\date{\today}
\begin{document}

\begin{abstract}
We introduce a correspondence between phylogenetic trees and Brauer diagrams, inspired by links between binary trees and matchings described by Diaconis and Holmes (1998).  This correspondence gives rise to a range of semigroup structures on the set of phylogenetic trees, and opens the prospect of many applications.  We furthermore extend the Diaconis-Holmes correspondence from binary trees to non-binary trees and to forests, showing for instance that the set of all forests is in bijection with the set of partitions of finite sets. 
\end{abstract}
\maketitle

\section{Introduction} 

Phylogenetic trees are a fundamental and persistent idea used to represent evolutionary relationships between species for nearly two centuries.  Their use extends beyond biology to the representation of language evolution, and to decision processes in algorithms.  Their appeal is that they provide an extra dimension to the ways to relate the elements of a set beyond a linear order.  They have been studied directly in numerous ways, through stochastic processes, combinatorics, and geometry (see for example~\cite{semple2003phylogenetics,steel2016phylogeny}).

An indirect but powerful way to study tree structures is to consider correspondences, or ways to represent trees, within other mathematical objects.  This is a standard approach in much of mathematics of course (representation theory is defined by this idea).  And there are several known bijections between rooted binary phylogenetic trees and other structures.

For instance, trees correspond to certain polynomials~\cite{liu2021tree,liu2020polynomial}, to perfect matchings~\cite{diaconis1998matchings}, and to more general partitions of finite sets~\cite{erdos1989applications}.  The latter has been extended to classes of forests relevant to phylogenetics via a correspondence between forests and trees, providing a way to enumerate such forests using Stirling numbers~\cite{erdos1993new}.  Other frameworks for capturing the combinatorics of tree counting and tree shapes have been through the analysis of binary sequences \cite{MR292690}, via symmetric function theory \cite{read_1968}, and numerous others (see for example the OEIS listing A001147~\cite{oeis}). 
The particular connection with perfect matchings has prompted the suggestion by Diaconis and Holmes~\cite{diaconis1998matchings} that there may be a relation to certain diagrams that have independently arisen within several branches of algebra, notably those given by the Brauer algebra~\cite{Brauer1937algebras}.

In this paper we take up and develop the link to Brauer algebras, and develop correspondences between the related diagram structures and phylogenetic trees, that provide an algebraic framework for their study.  
Each labelled, rooted binary tree will correspond to a unique element in the Brauer category, whose elements can be represented by asymmetrical Brauer diagrams.  More generally, a corresponding structure for non-binary trees can be defined using the associated partition category (for partition categories see for instance~\cite{martin1994temperley,jones1994potts}), and we are able to extend this correspondence to the set of all forests.  The extension to forests is associated with a new bijective correspondence between forests and partitions of finite sets that directly extends the ideas of Erd\H{o}s and Sz\'ekely~\cite{erdos1989applications} in a direction distinct from that in Erd\H{o}s~\cite{erdos1993new}.

The Brauer diagram framework has the potential to reveal more structure within the space of tree shapes (sometimes called topologies), and to provide new methods to randomly move about that space (something important for many tree reconstruction algorithms).  In particular, the inherited structure provides numerous ways to break down tree space: for instance, through the use of Green's relations from semigroup theory, we can break the set of phylogenetic trees into $\LL$-classes, $\RR$-classes, $\HH$-classes, and $\DD$-classes, each of which can be interpreted in the light of standard properties of trees (see Section~\ref{s:structure} for definitions of these semigroup classes). A concomitant of this is the possibility to craft new metrics on tree space.  Finally, we are able to define an operation on tree-space, effectively a multiplication of trees, \emph{relative to a fixed tree}, that turns the space into a semigroup, using the ``sandwich semigroup'' product. It is then possible to consider the ``regular'' elements of this semigroup, that provide a subsemigroup of the set of trees with respect to the chosen fixed tree. 

The paper begins in Section~\ref{s:background} with some background to the algebraic structures we will be using to describe tree space, namely the Brauer algebra and related monoids.  We also recall key concepts from semigroup theory, especially Green's relations and the corresponding equivalence classes into which a semigroup decomposes.  We then describe (Section ~\ref{ss:brauer}) the matching result of Diaconis and Holmes~\cite{diaconis1998matchings}, and show how it extends to a correspondence with Brauer diagrams.  Section~\ref{s:structure} then explores the result of transferring the results on the analysis and properties known in the Brauer category, across to phylogenetic trees, and discusses the rich new structure for tree space that becomes available.  Section~\ref{s:sandwich} introduces a product on trees, relative to a given tree, derived from the definition of the sandwich semigroup on Brauer diagrams, and explores various algebraic consequences of this powerful concept, including for example the characterization of related entities such as the associated regular subsemigroup of trees.  Finally, Section~\ref{s:non-binary.partitions} sketches a further broad generalisation of the correspondence, between non-binary trees, and associated partition diagrams and the partition category.  A key result is a correspondence results for trees and forests with sets of partitions of finite sets (Theorem~\ref{t:forests.partitions}).
We end with a discussion of directions that this algebraic landscape for phylogenetic trees might open up for further research (Section~\ref{s:discussion}).


\section{Background on phylogenetic trees, matchings, and semigroups}\label{s:background}

For a set $X$ of cardinality $n\ge 2$, a rooted (binary) phylogenetic $X$-tree $T$ is a graph with $n$ labelled, valence 1 leaf vertices labelled by $X$, together with $n\!-\!2$ unlabelled, valence three, internal vertices, and an additional unlabelled valence 2 root vertex.  Note that we do not consider the root to be an internal vertex.  
We will usually assume without loss of generality that $X={[}n{]} = \{1,2,\cdots, n\}$.  If $v$ is a vertex of $T$ we write $c(v)$ for the set of children of $v$: vertices that are the targets of edges whose source is $v$, and we say $v$ is the \emph{parent} of the vertices in $c(v)$. 
The cardinality of the set $\RP_n^{bin}$ of such leaf-labelled binary trees is $|\RP_n^{bin}|=(2n\!-\!3)!!$. This and other tree-related counting problems are well studied~\cite{MR292690,10.2307/2412810,erdos1989applications}, (see also \cite{read_1968,stanley1986enumerative,felsenstein2004inferring}).

In the latter part of this paper we will be working with trees that are not necessarily binary.  These satisfy the same properties as the binary $X$-trees, except that the internal vertices are not necessarily of valence 3 (for which we use the term non-binary trees\footnote{Note that the standard term ``non-binary trees'' includes the set of binary trees.}): instead they have in-degree 1, and out-degree \emph{at least} 2.  The set of all rooted phylogenetic $X$-trees is denoted $\RP_n$.  

In some contexts we also permit $|X|=1$, in which case the `trivial' tree is an isolated leaf labelled by the element of $X$, and which is also the root.  Note, the trivial tree is not an element of $\RP_n$ or $\RP_n^{bin}$, which are restricted to trees on $n\ge 2$ leaves.

An $X$-\emph{forest} is a set of trees whose leaf sets partition $X$.  The \emph{components} of a forest are the individual trees that make it up.  The \emph{trivial} forest on $X$ is the one in which all component trees are trivial; that is, a set of isolated leaves labeled by the elements of $X$.

A perfect matching on a set is a partition of the set into pairs, and the number of perfect matchings on $n$ elements ($n$ even) is $(n\!-\!1)!!$\,. This gives rise to a natural correspondence between rooted phylogenetic trees, and perfect matchings, 
on $2(n\!-\!1)$ elements (in this case, the leaf set $[n]$\,, augmented by
an additional $(n\!-\!2)$ elements, $[n\!+\!1, 2n\!-\!2] =
\{n\!+\!1, n\!+\!2,\cdots, 2n\!-\!3,2n\!-\!2\}$\,).  
This correspondence was formalized as a bijective map by Diaconis and Holmes \cite{diaconis1998matchings}, crystallizing that in Erd\H{o}s and Sz\'ekely~\cite{erdos1989applications}, and is described as follows.

Associating the leaves of $T$ with the first $n$ nodes of the perfect matching,
the initial set of pairings amongst all elements $[2n-2]=\{1,2,\cdots,2n-2\}$\, is simply generated from the cherries of the tree (leaves which share an immediate common ancestor), which are at most \smash{$\lfloor \textstyle{\frac 12}n\rfloor$} in number.
The next unmatched node, $n+1$\,, is assigned
to the internal tree vertex, both of whose child vertices are already labelled, one of which has the numerically lowest label. The next lowest available node label for any unmatched internal tree vertex or vertices is in turn assigned, amongst those with the already-labelled child vertices, to the one containing the numerically lowest child.
The process repeats until the node $2n\!-\!2$\,, corresponding to the last unlabelled internal tree vertex\,, is finally identified with its partner. An example, with a chord diagram representation (see below) of the node matchings, used in~\cite{diaconis1998matchings}, is given in Figure~\ref{fig:Circle6leafEx}.
This algorithm is formalised in pseudocode in Algorithm~\ref{alg:DH}.

\begin{algorithm}[ht]
\caption{Internal vertex labelling algorithm (Erd\H{o}s-Sz\'ekely, Diaconis-Holmes)}\label{alg:DH}
\begin{algorithmic}
\Require $T\in\RP_n^{bin}$ 
\Require $n \geq 1$
\State $m\gets 0$ \Comment{The number of labelled non-leaf vertices}
\State $I\gets$ set of unlabelled vertices in $T$ \Comment{Always non-empty (includes the root)}
\While{$|I| >1$}
	\State $C\gets$ $\{v\in I\mid c(v)\text{ are all labelled}\}$ \Comment{Always non-empty}
	\For{$v\in C$}
		\If{$c(v)$ have the lowest label among elements of $C$}
		    \State $v \gets$ labelled $n+m+1$
		    \State $m\gets m+1$
		    \State $I\gets I\setminus\{v\}$
		\EndIf 
	\EndFor
\EndWhile
\end{algorithmic}
\end{algorithm}

The correspondence with perfect matchings gives rise naturally to a different family of diagrams, namely those in the Brauer algebra ${\mathfrak B}_{n\!-\!1}$, which (unsurprisingly) has the same dimension, and the partial Brauer monoid $\mathfrak B_{n,n\!-\!2}$. In the next section we introduce these algebraic structures, and describe the correspondence.

A \emph{(set) partition} of a finite non-empty set $X$ is a set of pairwise disjoint subsets of $X$, whose union is $X$. 
If $\pi$ is a set partition of $X$, write $\ell(\pi)$ for the number of subsets in $\pi$, and write $|\pi|:=|X|$.
An \emph{(integer) partition} of a positive integer $n$ is a multiset of positive integers whose sum is $n$.  In this paper we will use ``partition'' to refer to \emph{set} partition unless otherwise noted.

A \emph{semigroup} is a set with an associative operation.  A simple example is the set of positive integers $\N^{>0}=\{1,2,3,\dots\}$ with the operation of addition.  If the semigroup has an identity element with respect to its operation, it is called a \emph{monoid}.  An example is the non-negative integers $\N:=\N^{\ge 0}$, in which the identity with respect to addition is of course 0.
Later in the paper we will introduce the semigroup notions of Green's relations and ``eggbox diagrams'' that represent these relations. For an introduction to some of these notions from semigroup theory, we recommend~\cite{howie1995fundamentals,everitt2021sympathetic}.

\section{Connections between trees, the Brauer algebra $\Br_{n-1}$, and Brauer diagrams $\Br_{n,n-2}$}
\label{ss:brauer}

In the first part of this paper we develop the correspondence between binary phylogenetic trees and perfect matchings via a further diagrammatic setting, that was already referred to in~\cite{diaconis1998matchings}: namely, by exploiting diagrams linked to the Brauer algebra.
These will share a basic structure and fundamental elements (the generators of the algebra), and so we briefly introduce the Brauer algebra ${\mathfrak B}_n(\lambda)$\, itself, before moving to the specific set of diagrams ${\mathfrak B}_{n,n-2}$ that will be our focus.

\subsection{Introduction to Brauer diagrams}\label{ss:brauer.intro}

For each $n=1,2,\cdots$ the Brauer algebra ${\mathfrak B}_n(\lambda)$ is the complex algebra generated by the elements~\cite{Brauer1937algebras} 
$$\{s_1,s_2,\cdots, s_{n\!-\!1}, e_1,e_2,\cdots, e_{n\!-\!1}\}$$
subject to the defining relations given in Appendix~\ref{s:brauer.relations}.  For the purposes of links to phylogenetics, 
we will make two significant restrictions to this generality.  

First, we will restrict to $\lambda=1$, and focus on $\Br_n=\Br_n(1)$ (although the potential to extend the ideas we discuss here to the full Brauer algebra makes an interesting question that we leave to future work).  
Second, we will treat ${\mathfrak B}_n$ and related objects ${\mathfrak B}_{m,n}$ (see below) as monoids (or partial monoids), working with just the basis elements of the algebra. These have a standard diagram transcription, in which elements are associated with 
graphs consisting of $2n$ nodes: `upper' nodes $[n]:=\{1,2,\cdots,n\}$ and `lower' nodes $[n']=\{1',2',\cdots,n'\}$ with edges specified as 
set partitions of $[n]\cup [n']$ consisting of pairs of nodes. In particular, $s_i$ contains the pairs $\{i,(i\!+\!1)'\}, \{(i\!+\!1), i'\}$ with all remaining nodes
$\{1, 1'\}, \{2,2'\}, \cdots $,  sequentially paired, while
$e_i$ contains $\{i,(i\!+\!1)\}, \{i',(i\!+\!1)'\}$\,, with all remaining nodes sequentially paired (see Figure \ref{fig:GeneratorsSandE}).

\begin{figure}[ht]
\centering
\includegraphics{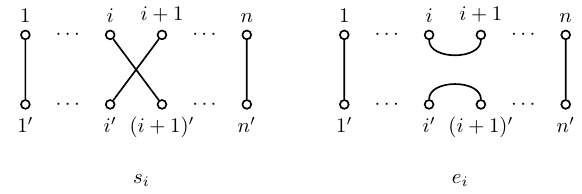}
\caption{Diagrammatic representation of the generators $s_i$ and $e_i \in \B_n$\,, for $i =1,2,\cdots,n\!-\!1$\,.  }
\label{fig:GeneratorsSandE}
\end{figure}

The (associative) product is formed by diagram concatenation, with edges joined via identification of the lower $[n']$ nodes of the first (upper) multiplicand, and the upper $[n]$ nodes of the second (lower) multiplicand. An example is shown in Figure~\ref{f:brauer.product.eg}.

\begin{figure}
    \centering
\includegraphics[width=.85\textwidth]{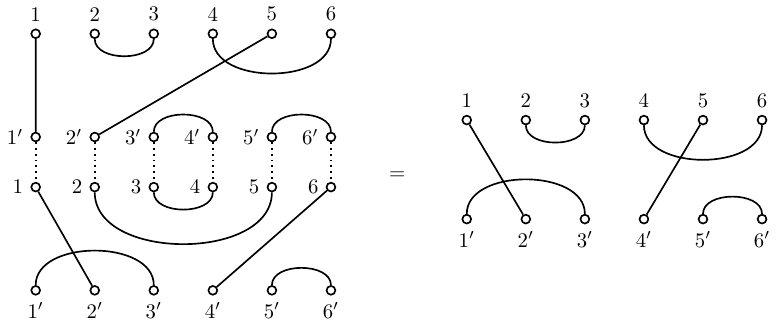}
    \caption{Multiplication of Brauer diagrams in $\Br_6$.}
    \label{f:brauer.product.eg}
\end{figure}

\subsection{Binary phylogenetic $X$-trees and their correspondence with Brauer diagrams}\label{ss:trees}

Given a binary tree on $n\ge 2$ leaves, once its internal vertices are labelled according to Algorithm~\ref{alg:DH}, we immediately have a matching on the set $\{1,\dots,2n-2\}$, obtained by pairing labels on sibling vertices~\cite{erdos1989applications,diaconis1998matchings}.  The reverse direction --- constructing a tree from a matching --- begins with laying out the leaf vertices and then pairing matched vertices that are already there, in increasing order.

Informally, the relation to matchings can be seen as the result of the growth of a binary tree via successive bifurcations: starting from a simple tree with two leaves (a single `cherry'), with matching $\{1,2\} =[2]$\,, one or the other edge suffers a bifurcation, giving for example $\{ \{1,2\}, \{21,22\}\}$ (a matching of $[4]$), whereupon
$2$, the parent (still paired to its original partner) becomes an internal node,
with the children $\{21,22\}$ a new pair of nodes. In this way, by iteration, a correspondence is built with matchings, or indeed Brauer diagrams, 
provided a unique convention for internal node enumeration is given
(tantamount simply to having `unlabelled' internal nodes), as has been described above following~\cite{diaconis1998matchings}. The root location is inferred from the last matched pairing.

Through the bijection between binary trees and matchings described in Section~\ref{s:background}, structural features of the Brauer
monoid are induced on trees.  Each binary tree on $n$ leaves uniquely defines a matching on $[2n-2]$, and each such matching uniquely defines a Brauer diagram in $\Br_{n-1}$ by connecting nodes labelled by integers paired in the matching (see example in Figure~\ref{fig:6leafEx}). Write $\delta(T)\in {\mathfrak B}_{n-1}$ for the Brauer diagram corresponding to $T\in\RP_n^{bin}$.

\begin{figure}[ht]
\includegraphics{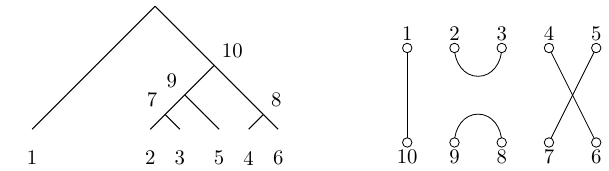}
\caption{The six-leaf tree $T\in \mathcal{RP}_6$ on the left corresponds to the matching $\{\{1,10\},\{2,3\},\{4,6\},\{5,7\},\{8,9\}\}$, giving the Brauer diagram $\delta(T)\in {\mathfrak B}_5$ shown on the right.}
\label{fig:6leafEx}
\end{figure}

If the product of two elements $\alpha, \alpha' \in {\mathfrak B}_{n-1}$ is written $\alpha \alpha'$ using juxtaposition, then there is an induced product on trees, defined by
\[
T\cdot T' = \delta^{-1}\big(\delta(T)\delta(T')\big)
\]
Similarly, the $*$-involution in ${\mathfrak B}_{n-1}$\,, defined on pairings by 
interchanging $i \rightarrow i' $\,, $i'\rightarrow i$\, for all $i$, or on diagrams by reflection about the horizontal axis, induces an involution on trees,
\[
T^* = \delta^{-1}\big(\delta(T)^*\big)\,.
\]
For example, the tree in Figure~\ref{fig:6leafEx} satisfies $T^*=T$, because its diagram is symmetric about the horizontal axis.

The bijection from trees of course extends to a variety of 
representations of the Brauer monoid or combinatorial objects tied thereto.
From the perspective of matchings of an even set, for example, it is more natural to label elements as $[1,2(n\!-\!1)]$ rather than marking them as $[n\!-\!1]\cup[(n\!-\!1)']$\, (as already done in establishing bijection $\delta$)\,, and a diagram correspondence could simply be established via arcs between
a linear arrangement of nodes $\{ 1,2,\cdots, 2(n\!-\!1) \}$\,. A less biased arrangement is a chord diagram, with arcs linking an even number
$2(n\!-\!1)$ of nodes arranged in a circle, as in~\cite{diaconis1998matchings}. Relative to a fixed ordering,
such a diagram represents a set of transpositions, an element of the symmetric group ${\mathfrak S}_{2(n-1)}$\,. Thus, the chord diagram in Figure \ref{fig:Circle6leafEx} corresponding to the six leaf tree of Figure~\ref{fig:6leafEx}, is obtained by bending up the ends of the lower rail of nodes 
$\{n,n\!+\!1, \cdots, 2(n\!-\!1)\}$ of the corresponding Brauer diagram, so that they join on to the upper rail of nodes 
$\{1,2,\cdots, (n\!-\!1)\}$\,. Associating a labelled tree with an element of the symmetric group in this way, as a product of transpositions coming from the matching, confers yet another possible multiplicative structure (exploited in the transposition distance for trees~\cite{valiente2005fast}). 

\begin{figure}[ht]
\includegraphics{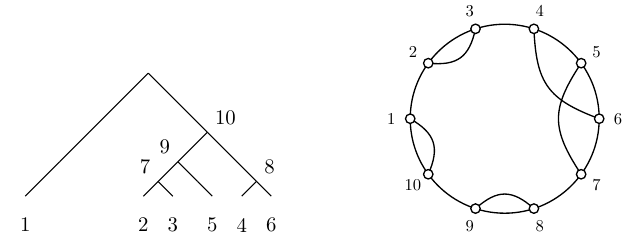}
\caption{The six-leaf tree $T\in \mathcal{RP}_6^{bin}$ from Figure~\ref{fig:6leafEx}, and its corresponding chord diagram representation.}
\label{fig:Circle6leafEx}
\end{figure}

For manipulations on trees, in this paper we will use a modified bijection, not between $\RP_n^{bin}$ and ${\mathfrak B}_{n-1}$\,, but between $\RP_n^{bin}$ and ${\mathfrak B}_{n,n-2}$\,, one of the equivalent diagrammatic presentations of matchings.  
In practice, the original $\delta$ and the modified bijection $\delta:\mathcal{RP}_n^{bin}\rightarrow{\mathfrak B}_{n,n-2}$ (for which we use the same symbol) are the same algorithmically, with the difference that the $n$ nodes on the upper edge correspond to the leaves of the $n$-leaf trees, and the $n-2$ nodes on the lower edge are those reserved for the corresponding internal tree nodes (as mentioned, with the root location being inferred from the last pairing). It is also convenient to continue the numbering from top to bottom in clockwise fashion. See Figure~\ref{fig:BrSixFourLfEx} for the Brauer transcription in ${\mathfrak B}_{6,4}$ of the six-leaf tree whose presentations in ${\mathfrak B}_{5}$\,, and as an element of ${\mathfrak S}_{10}$ (via a chord diagram) have been given in Figures~\ref{fig:6leafEx} and ~\ref{fig:Circle6leafEx}\,, respectively. 

\begin{figure}[ht]
\includegraphics{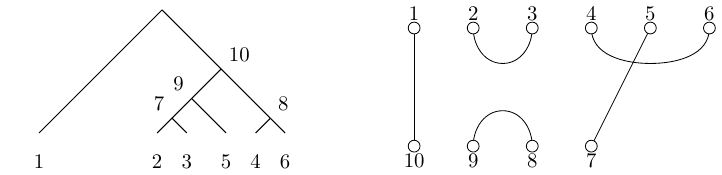}
\caption{The six-leaf tree $T\in \mathcal{RP}_6^{bin}$ from Figure~\ref{fig:6leafEx}, and its corresponding Brauer element $\alpha=\delta(T)\in {\mathfrak B}_{6,4}$\,. For this diagram we have dom$(\alpha)=\{1,5\}$, codom$(\alpha)=\{7,10\}$, $\ker(\alpha)=\{(2,3),(3,2),(4,6),(6,4)\}$, and $\coker(\alpha)=\{(8,9),(9,8)\}$, and rank$(\alpha)=2$.}
\label{fig:BrSixFourLfEx}
\label{f:kernel}
\end{figure}

With these conventions, the setting of trees in $\RP_n^{bin}$ as a semigroup, afforded by the transcription to ${\mathfrak B}_{n-1}$, is supplanted by 
the categorical setting of partial monoids \cite{MR3848011,MR3849129,dolinka2021sandwich}, where a multiplication on ${\mathfrak B}_{p,q} \times {\mathfrak B}_{r,s}$ exists only if $q=r$\, (compare \cite{morton1990knots}). In practice we 
analyze ${\mathfrak B}_{n,n-2}$\,, and hence $\mathcal{RP}_n^{bin}$\,, via  left- and right- actions by ${\mathfrak B}_{n}$, ${\mathfrak B}_{n-2}$ respectively, and the associated equivalence classes. 
Moreover, via multiplication with the help of intermediate `sandwich' elements belonging to ${\mathfrak B}_{n-2,n} \cong {\mathfrak B}^*_{n,n-2}$\,, a semigroup structure can indeed be re-imposed (see Section~\ref{s:sandwich}). This turns out to admit a rather universal description, allowing further structural aspects amongst the participating trees to be distinguished.

We now briefly introduce some language for describing features of a diagram $\dx\in\B_{n,n-2}$, that echoes that for functions, as follows.  
A \emph{block} of a diagram $\dx$ is a connected set of nodes in $\dx$.
A \emph{transversal} is an edge that passes between the top row and the bottom row.  
The \emph{domain} of a diagram $\dx$, $\dom(\dx)$, is the set of points along the top of the transversals, a subset of $[n]$. The \emph{rank} of $\dx$ is $|\dom(\dx)|$. The \emph{codomain}, $\codom(\dx)$, is the set of points along the bottom of the transversals.
The kernel $\ker(\dx)$ and cokernel $\coker(\dx)$ are defined slightly differently, for technical reasons which will become apparent later:
\begin{align*}
\ker(\dx)	&=\{(i,j)\in [n]\times[n]\,:\, \text{$i$ and $j$ belong to the same block}\}\\
\coker(\dx)	&=\{(i,j)\in [n+1,2n-2]\times[n+1,2n-2]\,:\, \text{$i$ and $j$ belong to the same block}\}.
\end{align*}
Note that if $\{i,j\}$ is a block in $\dx$ with $i,j\in [n]$, then both $(i,j)$ and $(j,i)$ are in $\ker(\dx)$.
An example is given in Figure~\ref{f:kernel}.  We will occasionally want to refer to the \emph{underlying set} of the kernel or cokernel.  By this we mean the set of all points in $[n]$ or $[n+1,2n-2]$ respectively that appear in relations in the sets.  That is, for a set of binary relations $S$, $U(S):=\{i\,:\, (i,j)\in S\text{ for some $j$}\}$.

With these definitions we note the following properties that hold for such a diagram $\dx\in\B_{n,n-2}$, that will be used in the sequel:
\begin{align*}
\dom(\dx)\cup U(\ker(\dx))=&\,[n] \qquad \text{and} \\
\codom(\dx)\cup U(\coker(\dx))=&\,[n+1,2n\!-\!2],
\end{align*}
noting that $\dom(\dx)\cap U(\ker(\dx))=\emptyset=\codom(\dx)\cap U(\coker(\dx))$,
and also the numerical relations
\begin{align*}
|\dom(\dx)|= &\,|\codom(\dx)|\quad\\
\text{and}\qquad |U(\ker(\dx))|=&\,|U(\coker(\dx))|+2\,.
\end{align*}

\section{Structuring tree space via Green's relations.}\label{s:structure}

With the identification with phylogenetic trees in $\mathcal{RP}_n^{bin}$\, as Brauer diagrams
of type ${\mathfrak B}_{n,n-2}$\, via the bijection $T \mapsto \delta(T)$\,, the correspondence with the monoidal structure in ${\mathfrak B}_{n-1} \equiv {\mathfrak B}_{n-1,n-1}$\,, afforded by the bijection $\delta$\,, is no longer direct. Rather, it is supplanted by 
the categorical setting of partial monoids \cite{MR3848011,MR3849129,dolinka2021sandwich}, where a multiplication on ${\mathfrak B}_{p,q} \times {\mathfrak B}_{r,s}$ exists only if $q=r$\,. 

In practice we analyze ${\mathfrak B}_{n,n-2}$\,, and hence $\mathcal{RP}_n^{bin}$\,, via  left- and right- multiplication by elements of ${\mathfrak B}_{n}$ and ${\mathfrak B}_{n-2}$ respectively, and the associated equivalence classes. 
{It is important to note here that the labels on the nodes that are matched in the product are implicitly re-assigned, along the lines of the description in Section~\ref{ss:brauer}\ref{ss:brauer.intro}, so that instead of bottom nodes reading $8,7,6$ from left to right in an element of $\Br_{5,3}$, we treat them as though labelled $1',2',3'$.  The multiplication on the bottom by an element of $\Br_3$ then matches $1',2',3'$ along the bottom of one diagram with $1,2,3$ along the top of the other.}

As in the previous usage, we pull the appropriate multiplications back to trees via the bijection:
\begin{align*}
\sigma \cdot T = &\, \delta^{-1}\big(\sigma \delta(T)\big)\,, 
\quad \mbox{for}\quad T \in \RP_n^{bin}\,, \sigma \in \B_{n}\,,\\
T\cdot \tau = &\, \delta^{-1}\big(\delta(T)\tau\big)\,,
\quad \mbox{for}\quad T \in \RP_n^{bin}\,, \tau \in \B_{n-2}\,.
\end{align*}

\begin{table}[tbp]
\includegraphics{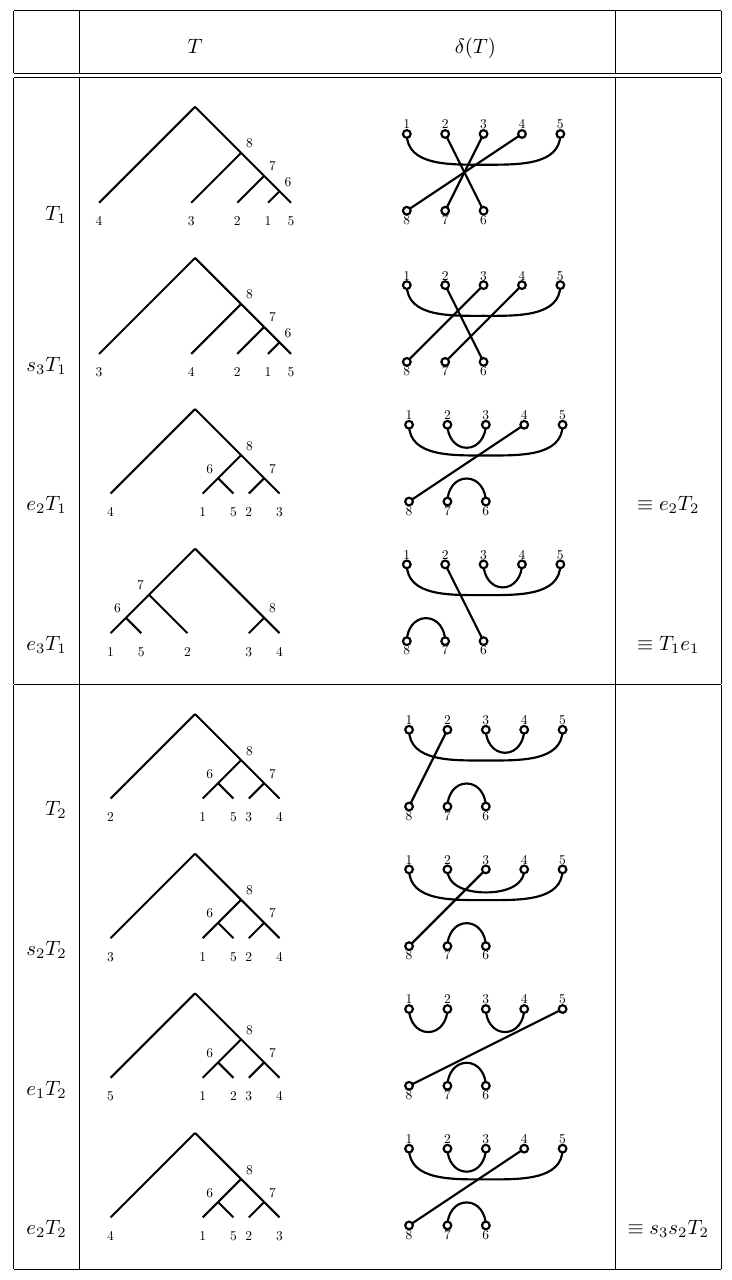}
\mbox{}\\
\caption{Examples of left- and right-actions of Brauer generators on 5-leaf trees in ${\mathfrak B}_{5,3}$\,.
Here $e_2$ has the same action on $T_2$ as the transposition $(24)=s_2s_3s_2$ (as well as the action of $s_3s_2$),
while $e_1$ acts identically to $(25)= s_2s_3s_4s_3s_2$. In this case, 
$e_2T_1= e_2 T_2$, and also $e_3T_1=T_1e_1$\,. 
}
\label{tab:fivethreeexx}
\end{table}
Examples of the results of multiplication of 5 leaf trees by Brauer generators are shown in Table~\ref{tab:fivethreeexx}. 

In semigroup theory, elements that can be obtained from each other by left- and/or right-multiplication are classified according to Green's relations (see~\cite{green1951structure,howie1995fundamentals}).   In the context of phylogenetic trees we find a \emph{restricted} version of these relations to be most useful, in which the actions are by elements of the symmetric group $\sS_n$ on the left (top of the diagram) and $\sS_{n-2}$ on the right (bottom), as opposed to the full Brauer monoid $\Br_{n}$ or $\Br_{n-2}$.  We define these restricted adaptions of Green's relations as follows:
\begin{defn}
Green's $\sL$, $\sR$, $\sD$, and $\sH$ relations are defined as follows, for $\B_{n,n-2}$:

\begin{tabular}{r@{:\quad}r@{\ $\Leftrightarrow$\ }ll}
$\sL$ classes & $T'\sL T$ & $\exists\,\alpha \in \sS_n : T' = \alpha\cdot T$, 
	& $[T]_{\sL} = \{ T': \exists\, \alpha \in \sS_n :T' = \alpha\cdot T \}$\\
$\sR$ classes & $T'\sR T$ & $\exists\, \alpha \in \sS_{n-2} : T' =  T\cdot\alpha$, 
	& $[T]_{\sR} = \{ T': \exists\, \alpha \in \sS_{n-2} :T' =  T \cdot\alpha \}$\\
$\sD$ classes & $T'\sD T$ & $T'\sL T \, \mbox{ or } \, T'\sR T$,
	& $[T]_{\sD} = [T]_{\sL} \cup [T]_{\sR}$\\
$\sH$ classes & $T'\sH T$ & $T'\sL T \,\mbox{ and } \,T'\sR T$,
	& $[T]_{\sH} = [T]_{\sL} \cap [T]_{\sR}$.
\end{tabular}
\end{defn}

In other words, the $\sL$ class containing $T$ in $\Br_{n,n-2}$ is the set of diagrams that can be reached from $T$ by a left multiplication (by an element of $\sS_n$).  Likewise, the $\sR$ classes arise from \emph{right} multiplication by elements of $\sS_{n-2}$.  The $\sH$ classes are sets of elements that are both $\sL$ and $\sR$ related, and these form the smallest of this family of equivalence classes.   In contrast, the $\sD$ classes are unions of intersecting $\sL$ and $\sR$ classes.  

These classes are displayed using arrays called \emph{eggbox diagrams} (see Chapter 2 of~\cite{howie1995fundamentals}).  Each $\sD$ class can be represented as a rectangular array whose entries are $\sH$-classes.  The $\sL$-classes are then given by the rows of the $\sD$-class, and the $\sR$-classes by the columns. 

Note that the actions by $\sS_n$ and $\sS_{n-2}$ are not faithful, because for instance the action of a transposition $s_i$ on a cup between $i$ and $i+1$ has the effect of the identity action (for example, in Table~\ref{tab:fivethreeexx}, $s_3T_2=T_2$).

Because top actions cannot affect the bottom of a diagram, all diagrams in an $\sL$ class have the same bottom half.  And since the action of $\sS_n$ on the top connects all arrangements of the top half of the diagram that have the same rank ($\sS_n$ is the full symmetric group), each $\sL$ class is the full set of diagrams with a particular bottom half.  

Likewise, each $\sR$ equivalence class is indexed by the common \emph{top} half of the diagrams within it.  The $\sH$ classes are then those diagrams that have the same top and bottom halves (the intersections of $\sL$ and $\sR$ classes), and there will be $k!$ of these, where $k$ is the rank of the diagrams in the class (this counts the number of ways to join the set of half-strings that come down from the top-half and up from the bottom half).  The $\sD$ classes are unions of $\sL$ and $\sR$ classes, namely all diagrams of the same rank.

Figure \ref{fig:EggboxEx} shows a Green eggbox scheme displaying Brauer diagrams (and corresponding trees) for the $\sD$ classes of the Brauer monoid $\Br_{6,4}\cong \RP_6^{bin} $ corresponding to 6 leaf trees. The class $\sD_2$ with two cherries (rank 2) is shown in detail, to illustrate the coordinatization of trees provided by the  Brauer structure.
This rank two $\sD_2$ class \cite{MR3848011,MR3849129,dolinka2021sandwich} and eggbox represents a total of 540 trees, with 6 rows ($\sL$ classes), and 45 columns ($\sR$ classes), with 270 intersections 
($\sH$ classes, of cardinality $2!$). The corresponding eggbox for rank 0
(three cherries), this eggbox (one cherry, rank 2), and the eggbox for rank 4 (caterpillar trees, with one cherry), together enumerate
the totality of $|\RP_6^{bin}|=9!! = 945 \equiv 45 +540  + 360$ labelled 6-leaf phylogenetic trees on 6 leaves, respectively.

\begin{figure}[ht]
\includegraphics[scale=1]{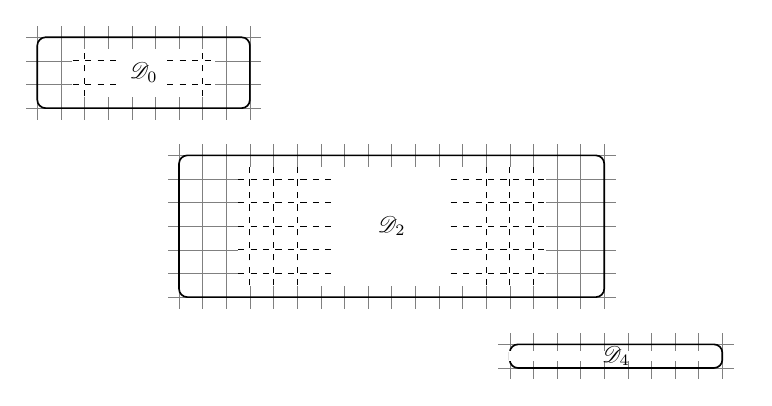}
\\
\includegraphics[width=\textwidth]{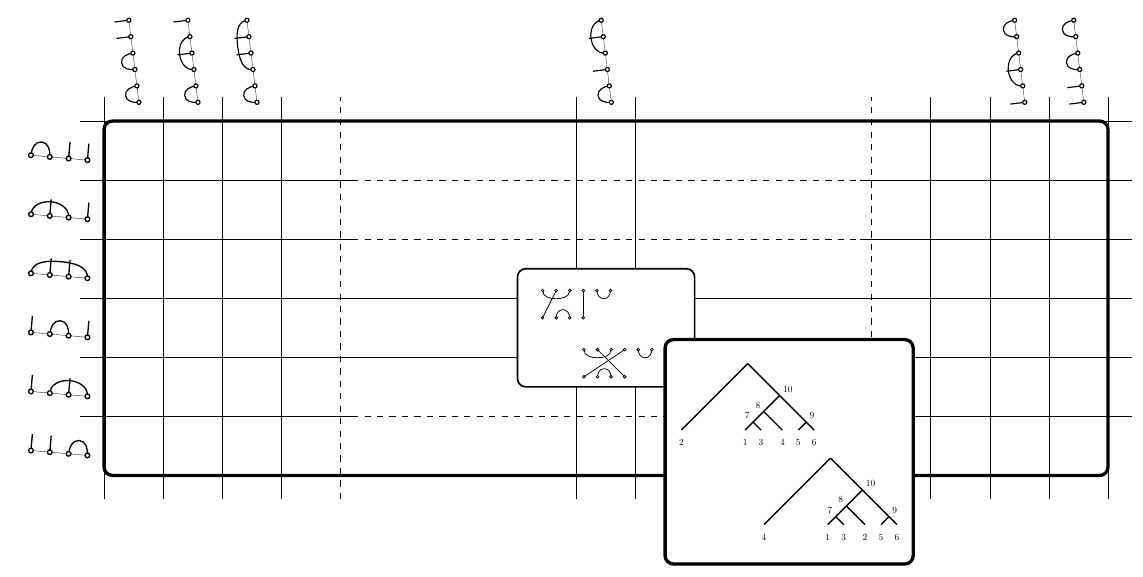}
\caption{An eggbox diagram illustrating the enumeration of elements of the ${\sD}$ classes of the Brauer monoid $\Br_{6,4}\cong \RP_6$, representing 6 leaf trees.
Upper diagram: schematic illustration of the separate eggboxes for the classes
${\sD_0}$, ${\sD_2}$  and ${\sD_4}$\,, of Brauer ranks $k=0,2,4$\,, corresponding to trees with $3, 2$ or $1$ cherry, and comprising $45, 540$ and $360$ trees, respectively (see text for discussion). Lower diagram: a ``close-up'' of the $\sD_2$ eggbox for trees with two cherries (and Brauer rank 2). 
Row labels show bottom-halves of Brauer diagrams, representing equivalence classes by left action (on the top of a diagram), the $\sL$-classes, whereas column labels show top-halves of diagrams, representing equivalence classes by right action, the $\sR$-classes. The inset shows two labelled trees, corresponding to the two diagrams belonging to the selected ${\sH}$ class (the intersection of the selected row and column). 
}
\label{fig:EggboxEx}
\end{figure}

\begin{rem}
These equivalence classes have natural interpretations in terms of trees. The $\sR$ classes are those trees that are the same up to permutations of leaf labels: their diagrams can be reached by multiplication by an element of $\sS_n$ along the top (which corresponds to the leaves).  The $\sL$ classes are those that have the same cherry structure at the leaves, but whose internal vertices have been permuted.  The $\sH$ classes represent trees with the same cherries but whose non-cherry leaves have been permuted.  Finally, the $\sD$ classes represent all trees with the same number of cherries (namely $\textstyle{\frac 12}(n-k)$, where $k$ is the rank).
\end{rem}


\section{Multiplicative structure: the set of phylogenetic trees as a semigroup}\label{s:sandwich}

In using a bijection between trees and an algebraic object like a Brauer monoid, the pay-off is the algebraic structure that comes to the set of trees.  In choosing to use the unbalanced diagrams of $\Br_{n,n-2}$, we preserve information about the tree structure (with the leaves all along the top axis of the diagram), but as noted above, we lose the capacity to multiply diagrams in the way that is possible if the top and bottom axes have the same number of nodes.   

\subsection{The sandwich product}

There is, nevertheless, still a multiplicative structure available to unbalanced diagrams such as $\Br_{n,n-2}$, namely the \emph{sandwich product}.  This requires a fixed diagram (tree) $T$, and then allows the product of two diagrams $T_1$ and $T_2$ to be composed by inverting $T$ (flipping it in the horizontal axis) to obtain a diagram $\oT\in \Br_{n-2,n}$, and sandwiching it between the two trees:
\[T_1\ast_T T_2:=T_1\cdot \oT\cdot T_2
\]
where $\cdot$ is composition of diagrams.  This product then allows us to define a semigroup relative to $T$, which we denote $\Br_{n,n-2}^T$.  
Figure \ref{fig:Sandwiching} illustrates such a sandwich product in terms of Brauer diagrams in $\Br_{7,5}$\, while the induced operation at the level of the corresponding trees is shown in Figure~\ref{fig:TREESandwiching}\,  (as can be seen, in this case the examples show a sandwich-square, of the form $T\ast_{T'} T := T \cdot \ot' \cdot T$\,)\,.

\begin{figure}[ht]

\begin{center}
\includegraphics{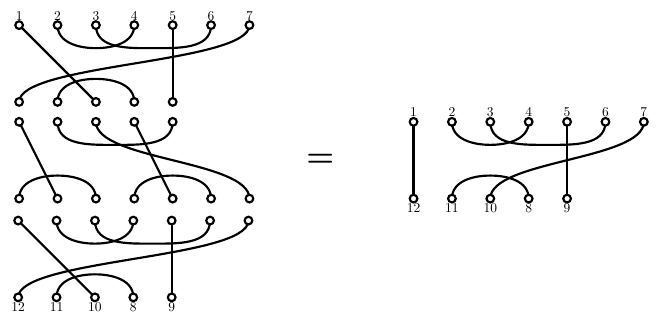}
\end{center}
\caption{The sandwich product of two Brauer diagrams in $\Br_{7,5}$, relative to a third (in the middle on the left). The corresponding trees are shown in Figure~\ref{fig:TREESandwiching}. }
\label{fig:Sandwiching}
\end{figure}

\begin{figure}
\begin{center}
\includegraphics[width=\textwidth]{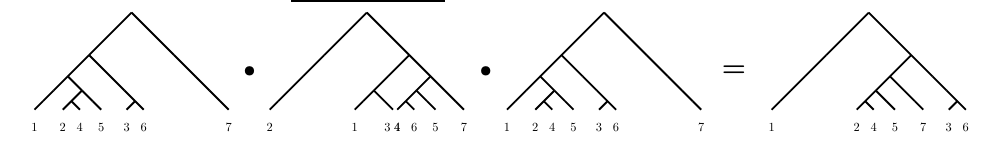}
\end{center}
\caption{The sandwich product of two trees relative to a third (with the bar over it). This product is computed using the diagram product of the corresponding Brauer diagrams, as shown in Fig~\ref{fig:Sandwiching}.}
\label{fig:TREESandwiching}
\end{figure}

Interestingly, the sandwich semigroups relative to different trees are isomorphic if the trees have the same \emph{rank} (the number of transversals in the diagram)~\cite[Theorem 6.4]{dolinka2021sandwich}.  This means that the choice of tree $T$ as sandwich is only important (in terms of semigroup structure) up to its rank.

Note that if a diagram has rank $k$, then its composition with any other diagram must have rank at most $k$, because composing diagrams cannot generate additional transversals.  In particular, the sandwich semigroup relative to tree $T$ does not contain an identity in general, because no sandwich product with a tree $T'$ of rank greater than rank$(T)$ can ever return a tree of the same rank as $T'$. 
As a consequence, the sandwich semigroup is not a monoid.

\subsection{The regular subsemigroup relative to a given tree}

Let $T\in\B_{n,n-2}$ be an $n$ leaf tree.  Using the sandwich product defined above, the set of trees with operation relative to $T$ is called the sandwich semigroup, denoted $\Br_{n,n-2}^{T}$.  The ``regular'' elements of this semigroup form a subsemigroup $Reg(\Br_{n,n-2}^{T})$.  A semigroup element $x$ is said to be \emph{regular} if there exists an $a$ such that $xax=x$. For the sandwich semigroup, this is equivalent to the property: $\alpha$ is regular if and only if $\oT\cdot\alpha\cdot\oT$ has the same rank as $\alpha$~\cite{dolinka2021sandwich}.  

The regular elements in this Brauer sandwich semigroup can be characterised as follows (we extract the case relevant to this context).

\begin{prop}
[\cite{dolinka2021sandwich} Prop 6.13]\label{p:regular.elts}
\begin{align*}\text{Reg}(\Br_{n,n-2}^{T})=\{\alpha\in\Br_{n,n-2}\ :\ &\coker(\alpha)\vee \ker(\oT) \text{ separates }\codom(\alpha)\\  &\text{ and }\ker(\alpha)\vee \coker(\oT) \text{ separates }\dom(\alpha) \}.
\end{align*}
\end{prop}

Here, the \emph{join} of two equivalence relations is their join in the lattice of equivalences, that is, the smallest equivalence relation that contains their union.  
The join $A\vee B$ \emph{separates} $C$ if each equivalence class in $A\vee B$ contains at most one element of $C$.

The elements of $\text{Reg}(\Br_{n,n-2}^{T})$ can be enumerated according to their rank relative to that of $T$, as follows:

\begin{thm}[\cite{dolinka2021sandwich} Corollary 6.17] \label{t:regular.size}
The cardinality of the regular subsemigroup of $\Br_{n,n-2}^{T}$ for any diagram $T$ of rank $k$, is given by
\[
|\text{Reg}(\Br_{n,n-2}^{T})|=\sum_{\substack{0\le j\le k\\ j\equiv k\mod 2}} \binom{k}{j}^2\frac{(k-j-1)!!^2(n+j-1)!!(n+j-3)!!j!}{(k+j-1)!!^2}
\]
\end{thm}

\begin{eg}\label{eg:reg.subsemi}
The tree $T$ whose diagram in $\Br_{6,4}$ is $\big\{\{1,10\},\{2,9\},\{3,4\},\{5,8\},\{6,7\}\big\}$ has rank $k=4$ (see Figure~\ref{f:diagram.tree.binary}). The sum in Theorem~\ref{t:regular.size} is over $j=0,2,4$, and can be computed as follows:
\begin{align*}
|Reg(\Br_{6,4}^{T})|&= \sum_{j=0,2,4} \binom{4}{j}^2\frac{(3-j)!!^2(5+j)!!(3+j)!!j!}{(3+j)!!^2}\\
&= \binom{4}{0}^2\,\frac{3!!^25!!3!!}{3!!^2} + \binom{4}{2}^2\,\frac{7!!5!!2!}{5!!^2} + \binom{4}{4}^2\,\frac{9!!7!!4!}{7!!^2}\\
&=45+504+216\\
&=765.
\end{align*}

It is interesting to note the number of elements of each rank: 45 of rank 0, 504 of rank 2, and 216 of rank 4.  The total number of diagrams of these ranks is respectively 45, 540, and 360.  In other words, this regular subsemigroup contains all trees of rank 0, 504/540 ($93\frac{1}{3}\%$) of rank 2, and 216/360 (60\%) of rank 4.

We can dig a little further into this counting using the conditions in Proposition~\ref{p:regular.elts}.
Given that $\ker\ot=\emptyset$ and $\coker\ot=\ker T=\{(3,4),(4,3)\}$, the condition for $\alpha$ to be in this subsemigroup are that both:
$\coker(\alpha)\vee \emptyset=\coker(\alpha)$  separates $\codom(\alpha)$; and $\ker(\alpha)\vee \{(3,4),(4,3)\}$ separates $\dom(\alpha)$.  The first of these is trivially satisfied since the underlying sets of $\codom$ and $\coker$ of $\alpha$ are disjoint (they partition the set $\{7,8,9,10\}$).  The second condition is more easily approached by considering when it will \emph{not} hold, namely when the set underlying $\ker(\alpha)\vee \{(3,4),(4,3)\}$ has two or more elements in common with $\dom(\alpha)$.  Since $\ker$ and $\dom$ have disjoint underlying sets, this forces $\{3,4\}\subseteq\dom(\alpha)$.

These diagrams must have domain of size 2 or 4, and some simple counting gives the number with domain size 2 as 36, and the number of domain size 4 as 144, for a total of 180 diagrams not in the regular subsemigroup.  This gives a total number of diagrams of $765+180=945$, which is indeed the number of rooted trees on 6 leaves (which is $(2n-3)!!$, and here $9!!=945$).

\begin{figure}[ht]
\includegraphics[scale=0.8]{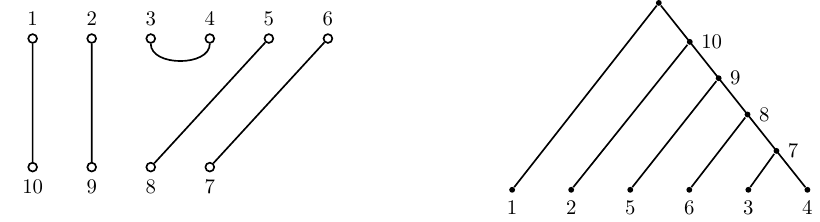}
\caption{The Brauer diagram and the tree from the partition $\big\{\{1,10\},\{2,9\},\{3,4\},\{5,8\},\{6,7\}\big\}$ in Example~\ref{eg:reg.subsemi}. }
\label{f:diagram.tree.binary}
\end{figure}

\end{eg}

There are several interesting questions related to the regular subsemigroup with respect to a tree, that we will leave for further work.  For instance, the regular subsemigroup of $T$ constitutes a type of neighbourhood of $T$ (noting that $T$ can be easily checked using Proposition~\ref{p:regular.elts} to be regular with respect to itself).  There are many ways to define a neighbourhood of a tree, for instance using operations on trees like the nearest neighbour interchange (NNI) and subtree prune and regraft (SPR) moves, or the transposition distance, also based on matchings~\cite{alberich2009algebraic,valiente2005fast}.  It would be interesting to know the relationships among these neighbourhoods.  Secondly, we can observe that in the case of Example~\ref{eg:reg.subsemi}, all trees of rank 0 are in the regular subsemigroup.  Is this a general property?  Are there properties of a tree that make it regular with respect to a large proportion of other trees?

\section{Non-binary trees and partition diagrams}
\label{s:non-binary.partitions}

In this section we extend the results from binary trees $\RP_n^{bin}$ to all trees $\RP_n$, and to forests, taking advantage of a more general family of diagrams and an associated algebraic structure, called a \emph{partition monoid}~\cite{dolinka2021sandwich}.  We begin with the generalisation to trees where the binary constraint is lifted (so that internal vertices and the root may have out-degree greater than 2).

\subsection{Non-binary trees}\label{s:nonbinary}

Recall that the underlying correspondence for binary trees on $n\ge 2$ leaves is that a tree corresponds to a matching on the set $\{1,2,\dots,2n-2\}$: a partition of the set of non-root vertices into components of size 2.  
The generalisation to non-binary trees maps a tree to a partition of the set of non-root vertices into \emph{subsets of size $\ge 2$} (note that the number of non-root vertices will be less than $2n-2$ if the tree is not binary).  For this reason, we will again exclude the trivial tree and require $n\ge 2$.

The generalisation begins by observing that Algorithm~\ref{alg:DH} applies without change when the input is a tree that is not necessarily binary (as in~\cite{erdos1989applications}).  We then obtain a partition from a tree by taking a subset to be a set of sibling vertices (having the same parent), and the correspondence immediately follows the same algorithm as for binary trees (see Figure~\ref{f:nonbinary.matching}).

Let $\Lambda_t$ denote the set of partitions of a set of $t$ objects, and $\Lambda_t^{(\ge 2)}$ the set of those partitions for which all constituent subsets have at least two elements.  We will also refer to $\Lambda_t^{(2)}$, the set of those partitions for which all components have size exactly two.  Note that $\Lambda_t^{(2)}$ is precisely the set of matchings on $t$ elements (in this case $t$ must be even).
Write $\Lambda=\cup_{t>0}\Lambda_t$ for the set of \emph{all} partitions of a finite set, and $\Lambda^{(\ge 2)}$ and $\Lambda^{(2)}$ for the corresponding sets when the sizes of components are at least 2 or exactly 2 respectively.

\begin{figure}[ht]
\includegraphics{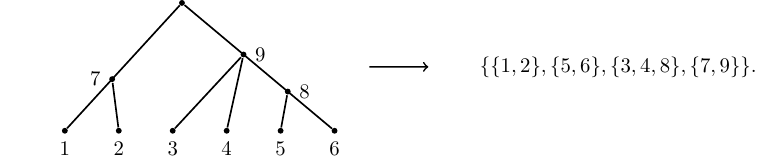}
\caption{A partition corresponding to a tree.  As with binary trees, non-leaf vertices in the tree are numbered in sequence using Algorithm~\ref{alg:DH}, choosing at each point the internal vertex whose children are all numbered and which has the lowest, numbered, child vertex.}\label{f:nonbinary.matching}
\end{figure}

We now introduce {partition diagrams}, which are generalisations of the Brauer diagrams defined above from partitions that are matchings to more general partitions.  

Recall that Brauer diagrams in $\Br_{n,n-2}$ have $2n-2$ nodes in two rows with $n$ nodes along the top numbered left to right 1 to $n$, and $n-2$ nodes along the bottom, numbered right to left $n+1$ to $2n-2$.   Nodes that are paired in the matching are connected by an edge.  A \emph{partition diagram} for an integer partition of $t=m+n$ has $n$ nodes along the top numbered from left to right 1 to $n$, and $m$ nodes along the bottom numbered right to left $m+1$ to $n+m$.  Nodes are connected by edges if their labels are in the same constituent subset of the partition.  If there are $k>2$ nodes in the subset, we do not draw all $\binom{k}{2}$ edges, but instead draw the minimal number to show their common membership, which will be $k-1$ edges.  Examples are shown in Figures~\ref{f:partition.tree} and~\ref{f:non-tree-diagram}.

Write $\dD_{n,m}$ for the set of partition diagrams with $n\ge 1$ nodes along the top and $m\ge 0$ along the bottom, and $\dD_{n,m}^{(\ge 2)}$ for the subset in which each partition is from $\Lambda_t^{(\ge 2)}$.   
The set $\Br_{n,n-2}$ is the subset of $\dD_{n,m}^{(\ge 2)}$ in which $m=n-2$ and all blocks have size exactly 2.  That is,  $\Br_{n,n-2}=\dD_{n,n-2}^{(2)}$.

For $T\in\RP_n$ (that is, not necessarily binary), write $\delta(T)$ for the corresponding element of $\dD_{n,m}^{(\ge 2)}$, and $\pi(\delta(T))$ for the corresponding set partition of $[m+n]$.  
Recall that $\ell(\pi)$ is the number of subsets in the partition $\pi$ of $X$, and $|\pi|$ is the cardinality of $X$.

The following lemma is the analogue of the property for diagrams from binary trees (that all satisfy $m=n-2$). 

\begin{lem}\label{l:m.and.blocks}
If $T\in\RP_n$, then the diagram $\delta=\delta(T)\in\dD_{n,m}^{(\ge 2)}$ satisfies $m=\ell(\pi(\delta))-1$.
\end{lem}

\begin{proof}
The blocks in the diagram $\delta$ correspond to sets of vertices in $T$ that have the same parent in $T$, therefore they are in one-to-one correspondence with the set of non-leaf vertices in $T$.  The non-leaf vertices in $T$ are represented by the $m$ numbered nodes along the bottom of the diagram, with the exception of the root of $T$.  Therefore, $m=\ell(\pi(\delta))-1$. 
\end{proof}
Note, this result means that given a partition $\pi$ of an integer $t=m+n$ with blocks of size at least 2, we can compute the values of $n$ and $m$ that give a tree corresponding to $\pi$.

\begin{eg}\label{eg:nonbinary.partition}
Consider the partition of a set of 12 elements given by 
\[\left\{\{1,3,12\},\{2,9\},\{4,6,8,11\},\{5,7,10\}\right\}.\]  
Since there are four blocks, Lemma~\ref{l:m.and.blocks} implies $m=3$.  The partition diagram and corresponding tree on nine leaves with three internal non-root vertices are shown in Figure~\ref{f:partition.tree}.
\end{eg}

\begin{figure}[ht]
\includegraphics[width=\textwidth]{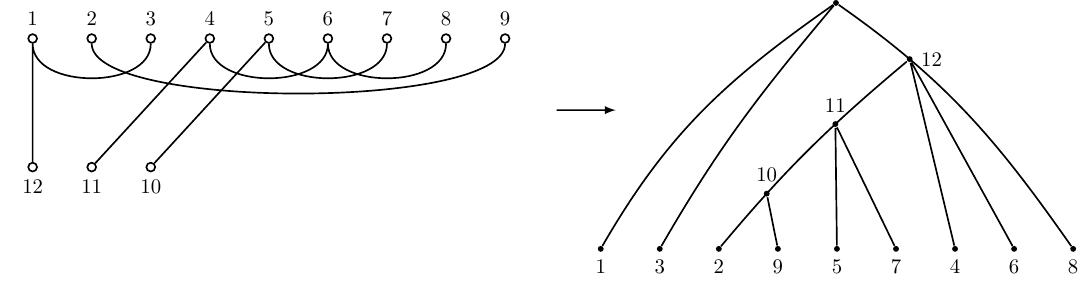}
\caption{Obtaining a non-binary tree from a partition with components of size at least 2, via a partition diagram. Here the partition is $\left\{\{1,3,12\},\{2,9\},\{4,6,8,11\},\{5,7,10\}\right\}$ given in Example~\ref{eg:nonbinary.partition}.}\label{f:partition.tree}
\end{figure}

While it is clear that each non-binary tree may be expressed as a partition diagram, it is not the case that every partition diagram is obtained from a tree, because some will not satisfy Lemma~\ref{l:m.and.blocks}.  The same observation holds, of course, for binary trees.

For example, the diagram in Figure~\ref{f:non-tree-diagram}(i) does not represent a tree.  However, the partition it displays, $\{\{1,2,6\},\{3,4,5\}\}$, corresponds to a tree via a diagram that we can find using Lemma~\ref{l:m.and.blocks} as follows (noting that $t=6$ and there are two blocks of size 3):
$n = 5$ and $m=1$.
The corresponding diagram and tree are shown in Figure~\ref{f:non-tree-diagram}(ii) and (iii).

\begin{figure}[ht]
\includegraphics[width=\textwidth]{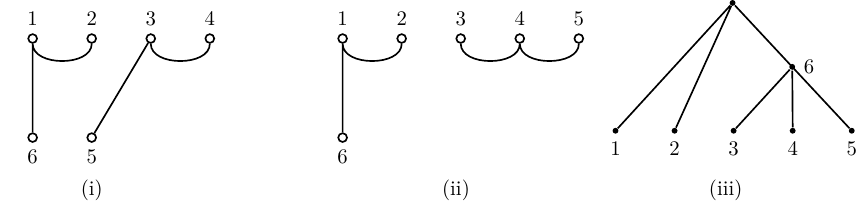}
\caption{(i) A diagram that does not correspond to a tree, because the corresponding partition $\pi=\{\{1,2,6\},\{3,4,5\}\}$ has $\ell(\pi)=2$ but the diagram does not have $m=\ell(\pi)-1=1$, as required by Lemma~\ref{l:m.and.blocks}.  (ii) The diagram of the same partition, but with the correct value of $m=1$. (iii) The corresponding tree.}
\label{f:non-tree-diagram}
\end{figure}

We are now able to generalise the correspondence between binary trees and matchings, to non-binary trees and sets of partitions.
Recalling that $\Lambda_t^{(\ge 2)}$ is the set of partitions of a set of $t$ elements into subsets of size $\ge 2$, let 
\begin{equation*}
\Lambda_{[n]}^{(\ge 2)} = \bigcup_{m=0}^{n-2}\left\{\pi\in\Lambda_{n+m}^{(\ge 2)}\mid\ell(\pi)-1=m\right\}.
\end{equation*}

Note that by using the fact that $m=|\pi|-n$, this may also be written 

\begin{equation}
\Lambda_{[n]}^{(\ge 2)} = \left\{\pi\in\Lambda^{(\ge 2)} \mid |\pi|-\ell(\pi)+1=n\right\}.
\label{e:Lambda.ge.2.no.m}
\end{equation}

\begin{thm}\label{t:trees.partitions.correspond}
There is a 1-1 correspondence between the set of partitions of finite sets into components of size $\ge 2$, and the set of (non-trivial) rooted phylogenetic trees.
\end{thm}

Note, the ``set of partitions of finite sets'' is not self-referential because the set of such partitions is infinite.

\begin{proof}
This result is a direct corollary to Theorem~\ref{t:forests.partitions} below.
\end{proof}

In each direction, a partition diagram may be constructed using the partition and the values of $m$ and $n$, so we also have as a consequence the following corollary.  Let
\begin{equation}\label{e:D[n].ge2}
\dD_{[n]}^{(\ge 2)}	:=\{\alpha\in\dD_{n,m}\mid \pi(\alpha)\in\Lambda_{[n]}^{(\ge 2)}\}.
\end{equation}

\begin{cor}\label{c:trees.diagrams.fixed.n}
The set $\RP_n$ is in bijection with the set of partition diagrams $\dD_{[n]}^{(\ge 2)}$.
\end{cor}

The correspondence in Theorem~\ref{t:trees.partitions.correspond} provides the potential for new ways to enumerate the set of rooted phylogenetic trees, by decomposing the set of partitions.  

For example, the set of all partitions of ordered sets into blocks of size $\ge 2$ is naturally sliced up according to the size of the ordered set, $t$.  In terms of trees on $n$ leaves, this groups them according their number of non-root vertices.  
In light of the above bijections, trees with particular characteristics such as this are able to be counted via the partial Bell polynomials
\cite{stanley1986enumerative} $B_{t,\ell}(x_1,x_2,\cdots, x_{t\!-\ell\!+\!1})$\,, whose monomial
coefficients count the number of set partitions $\pi$ of ${[}t{]}$\,, with $\ell(\pi)=\ell $\, blocks with specific frequencies.  Note that $t\!-\ell\!+\!1= |\pi|\!-\ell\!+\!1= n$, so that these are polynomials in $n$ indeterminates.
Thus the above sequence of trees sliced by the size of the ordered set, is given by \smash{$\sum_{\ell=1}^t B_{t,\ell}(0,1,\cdots, 1)$}\, whose first few terms are $1, 1, 4, 11, 41, 162,\cdots$\, (sequence A000296 of the On-Line Encyclopedia of Integer Sequences~\cite{oeis}), so that for instance there are $|\Lambda_5^{(\ge 2)}|=11$ partitions of a set of size 5 into partitions without singletons: $\binom{5}{2}$ ways to split into subsets of size 3 and 2 (trees on four leaves with one internal vertex), and 1 way to have a subset of size 5 (the star tree on five leaves).  Similarly,  the total number of trees with bifurcations or trifurcations only is the sequence
\smash{$\sum_{\ell=1}^t B_{t,\ell}(0,1,1,0,\cdots, 0)$}\, whose first few terms are
$1, 1, 3, 10, 25, 105, 385,\cdots$\, (sequence A227937 of~\cite{oeis}).

The former decomposition together with the correspondence in Corollary~\ref{c:trees.diagrams.fixed.n} can be represented in the diagram in Figure~\ref{f:decomposition.of.trees}.

\begin{figure}[ht]
\begin{tikzcd}
\Lambda_2^{(\ge 2)} \arrow[d] & \Lambda_3^{(\ge 2)} \arrow[d] & \Lambda_4^{(\ge 2)} \arrow[d] \arrow[dl] & \Lambda_5^{(\ge 2)} \arrow[d] \arrow[dl] & \Lambda_6^{(\ge 2)} \arrow[dl] \arrow[dll] & \Lambda_7^{(\ge 2)}  \arrow[dll] & \Lambda_8^{(\ge 2)}  \arrow[dlll] \\[6mm]
\Lambda_{[2]}^{(\ge 2)} \arrow[d,<->] & \Lambda_{[3]}^{(\ge 2)} \arrow[d,<->] & \Lambda_{[4]}^{(\ge 2)} \arrow[d,<->] & \Lambda_{[5]}^{(\ge 2)} \arrow[d,<->] & \phantom{\Lambda_{[2]}^{(\ge 2)}} &\phantom{\Lambda_{[2]}^{(\ge 2)}} & \phantom{\Lambda_{[2]}^{(\ge 2)}} \\[8mm] 
\RP_2 & \RP_3 & \RP_4 & \RP_5 &&&
\end{tikzcd}
\caption{The correspondence between sets of phylogenetic trees and sets of partitions described in Theorem~\ref{t:trees.partitions.correspond}, showing how the sets of partitions decompose the sets of trees.  An example of this decomposition for $\Lambda_{[5]}^{(\ge 2)}$ is shown in Figure~\ref{f:decomp.with.example.diags}. }
\label{f:decomposition.of.trees}
\end{figure}
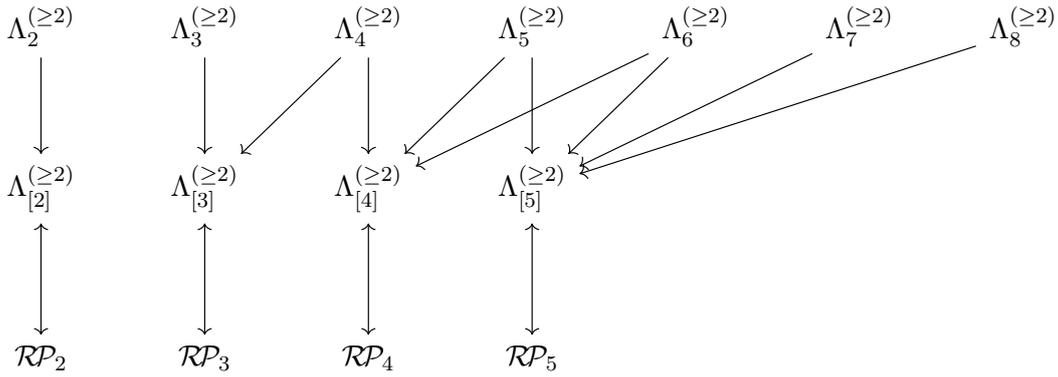

\begin{figure}[ht]
\begin{align*}
\Lambda_{[5]}^{(\ge 2)} :&= \bigcup_{m=0}^{3} \left\{\pi\in\Lambda_{5+m}^{(\ge 2)} \mid \ell(\pi)=m+1\right\} && \\
 &
=
\begin{minipage}{.4\textwidth}
  $ \left\{\pi\in\Lambda_{5}^{(\ge 2)} \mid \ell(\pi)=1\right\}$ 
\end{minipage} 
&& 
\begin{minipage}{.15\textwidth}
 \begin{tikzpicture}[scale=.5, odot/.style={circle,draw,fill=white,radius=2pt,inner sep=1.5pt,thick}] 
	\node (ref) {};
	\foreach \x in {1,...,5}  { \node[odot,left=10-3*\x mm of ref] (\x) {}; }
	\node[odot,below=3 mm of 1,white] () {};
	\draw (1) to[out=-90,in=-90] (2) to[out=-90,in=-90] (3) to[out=-90,in=-90] (4) to[out=-90,in=-90] (5);
 \end{tikzpicture}
\end{minipage} 
&& 
\begin{minipage}{.15\textwidth}
\includegraphics[height=6mm]{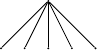}
\end{minipage} 
&&
\text{(1 tree)}
  \\
 &
 \begin{minipage}{.4\textwidth}
\quad $\cup \left\{\pi\in\Lambda_{6}^{(\ge 2)} \mid \ell(\pi)=2\right\} $
\end{minipage} 
&& 
\begin{minipage}{.15\textwidth}
 \begin{tikzpicture}[scale=.5, odot/.style={circle,draw,fill=white,radius=2pt,inner sep=1.5pt,thick}] 
	\node (ref) {};
	\foreach \x in {1,...,5}  { \node[odot,left=10-3*\x mm of ref] (\x) {}; }
	\foreach \x in {1}  { \node[odot,below=3 mm of \x] (b\x) {}; }
	\draw (1)--(b1);
	\draw (2) to[out=-90,in=-90] (3) to[out=-90,in=-90] (4) to[out=-90,in=-90] (5);
 \end{tikzpicture}
\end{minipage} 
&& 
\begin{minipage}{.15\textwidth}
\includegraphics[height=6mm]{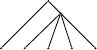}
\end{minipage} 
&&
\text{(25)}
 \\
 &
 \begin{minipage}{.4\textwidth}
 \quad $\cup \left\{\pi\in\Lambda_{7}^{(\ge 2)} \mid \ell(\pi)=3\right\}$ 
 \end{minipage} 
&& 
\begin{minipage}{.15\textwidth}
 \begin{tikzpicture}[scale=.5, odot/.style={circle,draw,fill=white,radius=2pt,inner sep=1.5pt,thick}] 
	\node (ref) {};
	\foreach \x in {1,...,5}  { \node[odot,left=10-3*\x mm of ref] (\x) {}; }
	\foreach \x in {1,2}  { \node[odot,below=3 mm of \x] (b\x) {}; }
	\draw (1)--(b1) (2)--(b2);
	\draw (3) to[out=-90,in=-90] (4) to[out=-90,in=-90] (5);
 \end{tikzpicture}
\end{minipage} 
&& 
\begin{minipage}{.15\textwidth}
\includegraphics[height=6mm]{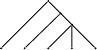}
\end{minipage} 
&&
\text{(105)}
 \\
 &
\begin{minipage}{.4\textwidth}
	\quad $\cup \left\{\pi\in\Lambda_{8}^{(\ge 2)} \mid \ell(\pi)=4\right\}$
\end{minipage} 
 &&
\begin{minipage}{.15\textwidth}
	\begin{tikzpicture}[scale=.5, odot/.style={circle,draw,fill=white,radius=2pt,inner sep=1.5pt,thick}] 
		\node (ref) {};
		\foreach \x in {1,...,5}  { \node[odot,left=10-3*\x mm of ref] (\x) {}; }
		\foreach \x in {1,...,3}  { \node[odot,below=3 mm of \x] (b\x) {}; }
		\draw (1)--(b1) (2)--(b2) (3)--(b3);
		\draw (4) to[out=-90,in=-90] (5);
	\end{tikzpicture}
\end{minipage}
&& 
\begin{minipage}{.15\textwidth}
	\includegraphics[height=6mm]{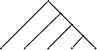}
\end{minipage}
&&
\text{(105)}
\end{align*} 
\caption{The decomposition of $\Lambda_{[5]}^{(\ge 2)}$ into sets of partitions. Examples of corresponding partition diagram shapes are in the centre column, with example corresponding tree shapes in the right hand column. The number of trees in each category is shown on the right, for instance in the second row there are $\binom{6}{2}+\frac{1}{2}\binom{6}{3}$ diagrams, and of course the last row is $(2(5)-3)!!$, giving the Ward numbers~\cite{ward1934representation} \cite[Seq. A269939]{oeis}.  Note, these are just examples and there are other possible diagram and tree structures with, for instance, two internal vertices (partitions of $[7]$ into three blocks).    This decomposes the set of all trees on 5 leaves according to the numbers of internal vertices, indicated by the number $m$ of nodes along the bottoms of the diagrams.}
\label{f:decomp.with.example.diags}
\end{figure}

\subsection{Forests}\label{s:forests}

The correspondence given in Section~\ref{s:non-binary.partitions}\ref{s:nonbinary} between trees and partitions applies to partitions with non-trivial subsets.  Recalling that subsets in a partition correspond to sibling vertices in a tree, a natural interpretation for a singleton (trivial) subset is that it corresponds to a vertex with no siblings.  Given the definition of a phylogenetic tree used here (and elsewhere) excludes non-root vertices of degree 2, the natural interpretation for a singleton subset is that it corresponds to a root vertex\footnote{Note that Erd\H{o}s and Sz\'ekely~\cite{erdos1989applications} do not strictly consider phylogenetic trees as we define them here, in that they allow internal vertices of degree 2 in their trees. Hence, their correspondence gives trees rather than forests.}.  And therefore, a diagram with singleton vertices ought to correspond to a forest.  Indeed, as in Theorem~\ref{t:forests.partitions}, forests of phylogenetic trees on $n$ leaves provide a one-to-one correspondence with a set of partitions $\Pi_{[n]}$ defined in Eq.~\eqref{e:Pi[n]} below.

Let $\F_n$ denote the set of $X$-forests, that is the set of forests whose leaves are labelled by elements of the set $X$, with $|X|=n\ge 1$.  $X$-forests are graphs whose connected components are rooted phylogenetic trees, whose leaves partition $X$.  Note that unlike the families of rooted trees, for forests we are allowing $n=1$.

Write $\ell_{\ge 2}(\pi)$ for the number of non-trivial blocks of the partition $\pi\in\dD_{n,m}$.  

Following the definition in Eq~\eqref{e:Lambda.ge.2.no.m}, define 
\begin{align}
\Lambda_{[n]}	&:=\{\pi\mid |\pi|-\ell_{\ge 2}(\pi)+1=n\}\label{e:Lambda[n]}\\
\dD_{[n]}		&:=\{\alpha\in\dD_{n,m}\mid \pi(\alpha)\in\Lambda_{[n]}\}.\label{e:Pi[n]}
\end{align}

Note, in $\dD_{[n]}$, $m=\ell_{\ge 2}(\pi(\alpha))-1$.  Let $\F_n$ denote the set of all forests on $n$ leaves, and $\F=\cup_{n\ge 1}\F_n$ the set of all forests.

\begin{thm}\label{t:forests.partitions}\ 

\begin{enumerate}
	\item $\F$ is in bijection with the set of partitions $\Lambda$;
	\item $\F_n$ is in bijection with the set of partitions $\Lambda_{[n]}$; and
	\item The set of non-trivial forests on $n$ leaves with $\tau$ components is in bijection with the set $\dD_{n,m}$ of partition diagrams with $\tau-1$ singleton nodes that satisfy $m=\ell(\pi)-\tau$.
\end{enumerate}
\end{thm}

\begin{proof}

(1) 
If $F\in\F_n$ is trivial, so that all its trees are isolated leaves, it will map to the trivial partition of $[n]$, that is, $\{\{1\},\dots,\{n\}\}$.  So we need to prove the correspondence between non-trivial forests and non-trivial partitions.

Each (non-trivial) forest gives a partition, by first numbering non-leaf vertices according to Algorithm~\ref{alg:DH}, and then forming sets of sibling vertices, with labelled root vertices forming singletons.  As with trees, this algorithm leaves a single root vertex un-labelled.  

We now explore the properties of the partition arising from a forest, to help in constructing the map back from partitions to forests.  Suppose $\pi=\pi(F)$ is the partition obtained from the forest $F$.

If $F$ is non-trivial, we have a correspondence between non-leaf vertices in $F$, and non-trivial blocks, given by the children of each non-leaf vertex.  It follows that the number of non-leaf vertices in $F$ is precisely $\ell_{\ge 2}(\pi)$.  

The number of all vertices in a non-trivial forest $F$ is $|\pi|+1$, because one vertex (one of the roots) is left unlabelled by Algorithm~\ref{alg:DH}.  The vertices in $F$ are also either leaves or non-leaves, and so this number is also equal to $n+\ell_{\ge 2}(\pi)$.  Therefore,
\[
|\pi|+1 = n+\ell_{\ge 2}(\pi),
\]
and so 
\(
n= |\pi| -\ell_{\ge 2}(\pi)+1.
\)

Now consider a non-trivial partition $\pi\in\Lambda$. We will show how a forest can be constructed from $\pi$.  Set $n= |\pi| -\ell_{\ge 2}(\pi)+1$, and create a starting forest $F_0$ consisting of $n$ vertices as leaves, labelled $1,\dots,n$.  
We will successively add vertices and edges to the forest as follows.

First, set $\pi_0=\pi$, and consider the non-trivial sets in $\pi_0$.  We have assumed that $\pi_0=\pi$ is non-trivial, so there will be at least one.  We claim that at least one of these is contained in $\{1,\dots,n\}$.  Observe that there are $|\pi_0|-n=\ell_{\ge 2}(\pi_0)-1$ integers outside  $\{1,\dots,n\}$ and included in $\pi_0$.  But there are $\ell_{\ge 2}(\pi_0)$ non-trivial subsets, and so at least one cannot contain an element outside $\{1,\dots,n\}$, as required.  

Let $\hat\pi_0$ be the set of non-trivial subsets in $\pi_0$ contained in $\{1,\dots,n\}$ (that label vertices in $F_0$), and let $S_0$ be the element of $\hat\pi_0$ containing the least integer.  This is well-defined because, as argued, $\hat\pi_0$ is non-empty.

Add a vertex $v_1$ to $F_0$, as a parent to the vertices labelled by the elements of $S_0$, to create a new forest $F_1$.  Then remove $S_0$ from $\pi_0$ to define 
\[\pi_1:=\pi_0\setminus S_0.\]
If $\pi_1$ has no non-trivial subsets, then end the algorithm and output $F_1$.  (In this case, we will have had $\ell_{\ge 2}(\pi)=1$, and so $|\pi|=n$ and all integers in $\pi$ are labelling vertices in $F_1$).

Otherwise, label the vertex $v_1$ in $F_1$ by $n+1$.  $F_1$ has leaves labelled $1,\dots, n$ and one other vertex labelled $n+1$ which is the parent of at least two of the leaves.

As before, we claim that $\pi_1$ contains a set that is a subset of $\{1,\dots,n,n+1\}$, and the argument naturally extends as follows:
\begin{itemize}
   	\item $\pi_1$ has $\ell_{\ge 2}(\pi_1)=\ell_{\ge 2}(\pi)-1\neq 0$ non-trivial subsets;
   	\item There are $|\pi|-(n+1)=\ell_{\ge 2}(\pi)-2$ integers in $\pi$ outside of $\{1,\dots,n+1\}$; and
   	\item Therefore it is not possible for all non-trivial subsets in $\pi_1$ to include an element outside $\{1,\dots,n+1\}$.
\end{itemize}   
Thus, the set $\hat\pi_1$ of non-trivial subsets in $\pi_1$ contained in $\{1,\dots,n+1\}$, is non-empty.  Choose the subset in $\hat\pi_1$ with the least integer, and call it $S_1$.

This process can continue, as described in Algorithm~\ref{alg:forest}, until we reach a point where $\pi_i$ has no non-trivial subsets and we output the resulting forest.  The forest will have $n$ leaves, and $i=\ell_{\ge 2}(\pi)$ additional vertices, of which $i-1$ will have labels: one root vertex will remain unlabelled.  Thus all elements of $\pi$ will be labelling vertices, since $|\pi|=n+\ell_{\ge 2}(\pi)-1$.

Note that singletons in the partition $\pi$ are also labelling vertices in the forest.  Any singleton in $\pi$ that is $\le n$ is already labelling an isolated leaf from the outset, and so represents a trivial tree.  Any singleton greater than $n$ will be labelling a root vertex in a tree in the forest, because it will be assigned as a parent of vertices in a block, and will not be assigned a parent because it has no other elements in its block.  See Example~\ref{e:forest.partition} for an illustration of these observations.

The process deterministically defines a forest on $n$ leaves whose vertices (except one root) are labelled by the elements of the partition $\pi$, and completes the proof of (1).

(2) immediately follows from the construction described above, which gives a correspondence between a forest with $n$ leaves and a partition satisfying the condition to be an element of $\Lambda_{[n]}$.

For (3), suppose $F$ has $n$ leaves and $\tau$ component trees, and is non-trivial.  Labelling the vertices according to Algorithm~\ref{alg:DH}, all but one of these trees will have a labelled root, and so the corresponding partition will have $\tau-1$ singletons.  As noted above, the number of non-leaf vertices is $\ell_{\ge 2}(\pi)$, because they correspond to sets of siblings, which correspond to blocks of the partition.  This includes the single non-labelled root vertex, and so the number of non-leaf labels in the partition for $F$ is $\ell_{\ge 2}(\pi)-1$, and this is the value of $m$ in the partition diagram. But since the partition has $\tau-1$ singleton sets, we have $\ell(\pi)=\ell_{\ge 2}(\pi)+(\tau-1)$, and it follows that $m=\ell(\pi)-\tau$ as required.   The reverse direction takes a diagram to a partition that then constructs a forest according to Algorithm~\ref{alg:forest}.  The conditions in the statement follow.
\end{proof}

\begin{algorithm}[ht]
\caption{Construction of a forest from a partition}\label{alg:forest}
\begin{algorithmic}
\Require $\pi\in\Lambda$ 
	\State $n\gets |\pi|-\ell_{\ge 2}(\pi)-1$
	\State $F_0\gets$ the trivial forest on leaves $\{1,\dots,n\}$
	\State $\pi_0\gets \pi$
	\State $\widehat\pi_0\gets$ the subsets of $\pi_0$ of size $\ge 2$ whose elements are labels in $F_0$
	\State $S_0\gets$ the subset in $\widehat\pi_0$ containing the least integer

	\For{$i>0$}
	    \State Create a vertex $v_i$ and edges $E_i:=\{(v_i,s)\}$ for vertices $s$ in $F_{i-1}$ labelled by $S_{i-1}$
		\State $F_i\gets F_{i-1}$ with the additional vertex $v_i$ and edges $E_i$ 
		\State $\pi_i\gets \pi_{i-1}\setminus S_{i-1}$
		\If{$\pi_i$ contains a non-trivial subset} 
			\State Label $v_i$ by $n+i$.
		\Else
			\State $F\gets F_i$
			\State Break \Comment{Go to Output statement}
		\EndIf
		\State $\widehat\pi_i\gets $ subsets in $\pi_i$ of size $\ge 2$ labelling vertices in $F_i$.
		\State $S_i\gets$ the subset in $\hat\pi_i$ with the lowest integer
	\EndFor
\State Output $F$.	

\end{algorithmic}
\end{algorithm}

\begin{eg}[Construction of a forest from a partition via Algorithm~\ref{alg:forest}.]
\label{e:forest.partition}

Take the partition $\pi$ given by:
\[\pi=\{\{1,10\},\{3,4,6\},\{5\},\{7,8\},\{9,13\},\{2,12\},\{11\}\}.
\]
Here we have $|\pi|=13$, $\ell_{\ge 2}(\pi)=5$, and $n=|\pi|-\ell_{\ge 2}(\pi)+1=9$.
Set $F_0$ to be the forest with 9 isolated leaves labelled $1,\dots,9$, and set $\pi_0=\pi$. We have $\hat\pi_0=\{\{3,4,6\},\{7,8\}\}$, and so $S_0=\{3,4,6\}$.

Let $F_1$ be the forest obtained from $F_0$ by adding a vertex $v_1$ that is a parent of the vertices labelled 3,4,6.  Set $\pi_1=\pi_0\setminus S_0 = \{\{1,10\},\{5\},\{7,8\},\{9,13\},\{2,12\},\{11\}\}$.  Since $\pi_1$ has non-trivial subsets we label $v_1$ by 10, and $\hat\pi_1=\{\{1,10\},\{7,8\}\}$, with $S_1=\{1,10\}$.

Let $F_2$ be the forest obtained from $F_1$ by adding a vertex $v_2$ that is a parent of the vertices labelled 1,10.  Set $\pi_2=\pi_1\setminus S_1 = \{\{5\},\{7,8\},\{9,13\},\{2,12\},\{11\}\}$.  Since $\pi_2$ has non-trivial subsets we label $v_2$ by 11, and $\hat\pi_2=\{\{7,8\}\}$, with $S_2=\{7,8\}$.

Let $F_3$ be the forest obtained from $F_2$ by adding a vertex $v_3$ that is a parent of the vertices labelled 7,8.  Set $\pi_3=\pi_2\setminus S_2 = \{\{5\},\{9,13\},\{2,12\},\{11\}\}$.  Since $\pi_3$ has non-trivial subsets we label $v_3$ by 12, and $\hat\pi_3=\{\{2,12\}\}$, with $S_3=\{2,12\}$.

Let $F_4$ be the forest obtained from $F_3$ by adding a vertex $v_4$ that is a parent of the vertices labelled 2,12.  Set $\pi_4=\pi_3\setminus S_3 = \{\{5\},\{9,13\},\{11\}\}$.  Since $\pi_4$ has non-trivial subsets we label $v_4$ by 13, and $\hat\pi_4=\{\{9,13\}\}$, with $S_4=\{9,13\}$.

Let $F_5$ be the forest obtained from $F_4$ by adding a vertex $v_5$ that is a parent of the vertices labelled 9,13.  Set $\pi_5=\pi_4\setminus S_4 = \{\{5\},\{11\}\}$.  Since $\pi_5$ has no non-trivial subsets we end the algorithm and output $F=F_5$.

The forests generated in this example are shown in Figure~\ref{f:forests.algorithm}.
\end{eg}

\begin{figure}[ht]
\includegraphics{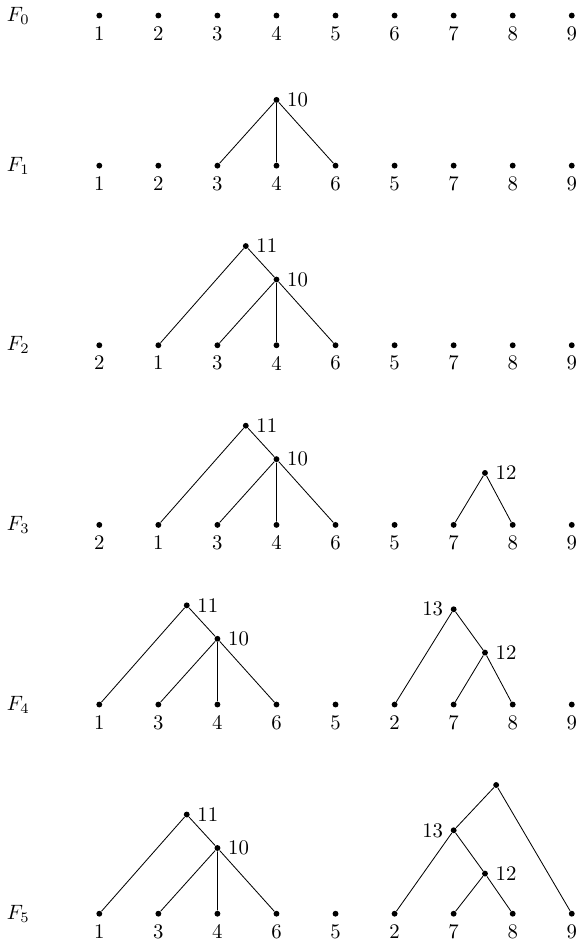}
\caption{The forests generated using Algorithm~\ref{alg:forest}, to obtain a forest from the partition $\pi=\{\{1,10\},\{3,4,6\},\{5\},\{7,8\},\{9,13\},\{2,12\},\{11\}\}$, as described in Example~\ref{e:forest.partition}, outputting $F=F_5$.}
\label{f:forests.algorithm}
\end{figure}

Note, Theorem~\ref{t:trees.partitions.correspond}, relating to non-binary trees, is a special case of this theorem for those forests that consist of a single tree, and partitions without singletons. 

The correspondence with forests given in Theorem~\ref{t:forests.partitions} creates a broad set of correspondences between sets of partitions and sets of phylogenetic objects, as shown in Figure~\ref{f:correspondences}. Recall that we have defined the following sets:
\begin{align*}
\Lambda_n^\star &= \parbox[t]{12cm}{\raggedright Set of partitions of a set of cardinality $n$, whose components satisfy the size condition $\star$.}\\
\dD_{n,m}^\star &= (n,m) \text{ partition diagrams from partitions in }\Lambda_{n+m}^\star.\\
\Lambda_{[n]}^\star &= \left\{\pi\in\Lambda^\star\ \big\vert\ |\pi|-\ell(\pi)+1=n\right\}.\\
\dD_{[n]}^\star &= \left\{\alpha\in\dD_{n,|\pi|-n}^\star\ \big\vert\ \pi(\alpha)\in\Lambda_{[n]}^\star\right\}.
\end{align*}
Here $\star$ might be empty (no restrictions), or $(2)$ or $(\ge 2)$, meaning components must be of size 2 or at least 2.  But in the next two corollaries we see correspondences for when components have size \emph{at most} 2.

\begin{figure}[ht]
\begin{tikzcd}
|[left]| \text{phylogenetics:} & \RP^{bin}_n \arrow[r, "\subset"] \arrow[d,<->]	 & \RP_n \arrow[r, "\subset"] \arrow[d,<->]  	& \F_n 	\arrow[d,<->] 	\\
|[left]| \text{partitions:} & \Lambda_{2n-2}^{(2)} \arrow[r, "\subset"] \arrow[d,<->] & \Lambda_{[n]}^{(\ge 2)} \arrow[r, "\subset"] \arrow[d,<->] & \Lambda_{[n]} \arrow[d,<->] \\
|[left]| \text{diagrams:} & \Br_{n,n-2} \arrow[r, "\subset"]    & \dD_{[n]}^{(\ge 2)} \arrow[r, "\subset"] & \dD_{[n]} \\
& |[left,rotate=90]| \text{binary trees:} & |[left,rotate=90]| \text{all trees:} & |[left,rotate=90]| \text{forests:}  
\end{tikzcd}
\caption{Correspondences between sets of trees or forests, all on $n$ leaves, sets of partitions, and partition diagrams. 
In the left column, the components of the partition are all size exactly 2, and so $m=n-2$.  
}
\label{f:correspondences}
\end{figure}
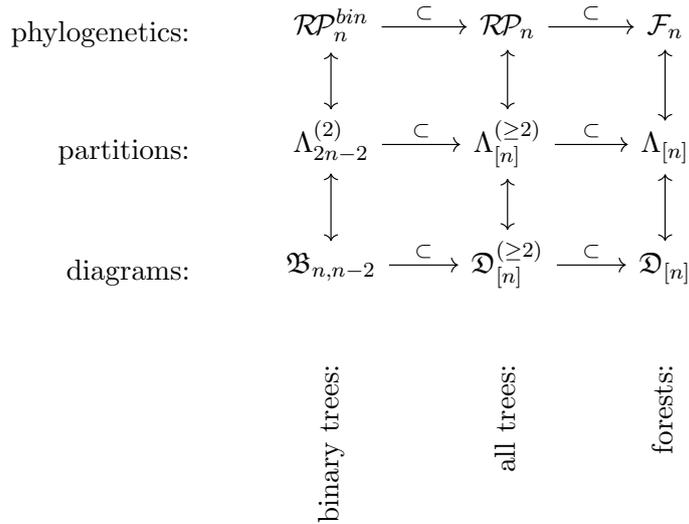

\begin{cor}\label{c:binary.forests}
The set of binary forests is in bijection with the set of partitions $\Lambda^{(\le 2)}$: those whose subsets have size at most 2. \end{cor}

\begin{cor}\label{c:binary.forests.n.leaves}
The set of binary forests on $n$ leaves is in bijection with the set of partitions $\Lambda_{[n]}^{(\le 2)}$.
\end{cor}

\subsection{Semigroup structure}
The semigroup structures that we have described for binary trees in Section~\ref{s:sandwich}, also extend to non-binary trees or forests, with some caveats.  For instance, immediately we see that an action by the Temperley-Lieb generators $e_i$ will not be able to be defined for non-binary trees in general: the product of $e_i$ by a partition diagram that includes a block $\{i,i+1,j\}$ where $j$ is on the opposite side to $i$ and $i+1$, will result in a block of size one, namely $\{j\}$.  So the action of a Temperley-Lieb generator on a non-binary tree may result in a forest.

However, action by the symmetric group generators $s_i$ does preserve the restriction on the partition diagram.  In fact, action by $s_i$ preserves the number and size of the blocks in the partition.  Therefore, we are still able to construct the eggbox diagram decomposition of the set of non-binary trees, and indeed for forests.  This decomposition will have a richer structure, however, as the $\DD$-classes are not determined simply by rank, but by other factors (the number and size of the blocks in the partition).  This is a topic we leave for further investigation.

The sandwich semigroup construction allows multiplication of non-binary trees, as it does in the binary case, although in general the product will not be closed. That is, if one of the trees involved in the sandwich product is not binary, it is possible that the product of diagrams results in an isolated node, which means the diagram may correspond to a forest of more than one tree (see Figure~\ref{f:sandwich.non.tree}).
Thus, the set of diagrams for non-binary trees is not closed under the sandwich product.
Whether there are subfamilies of non-binary trees for which the product is defined is an interesting further question.

\begin{figure}[ht]
\includegraphics[width=\textwidth]{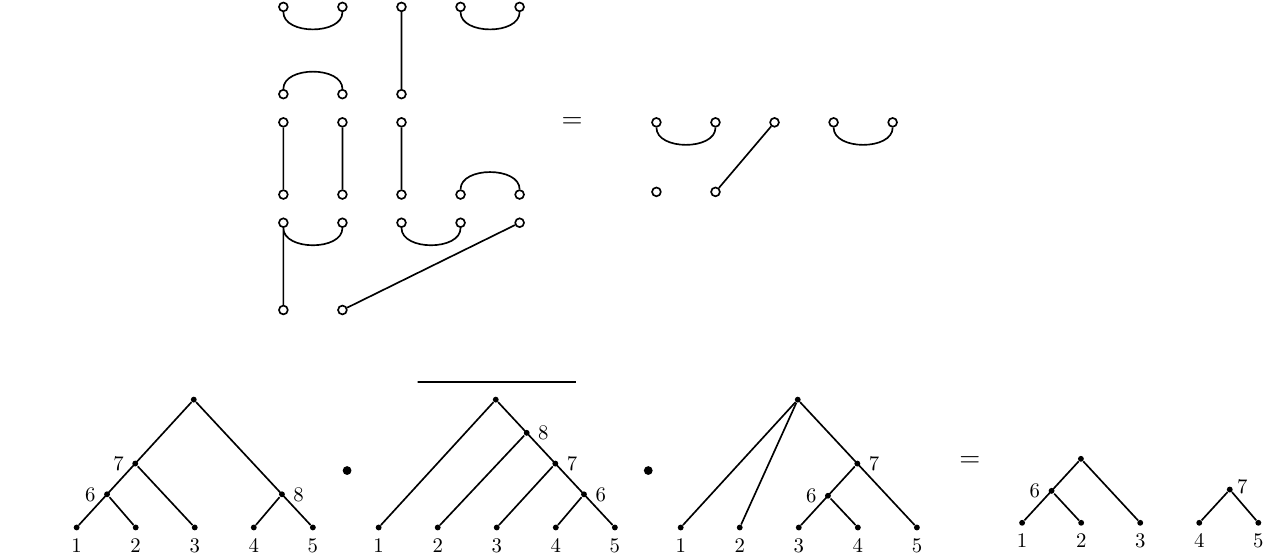}
\caption{Top: A sandwich product involving a diagram from a non-binary tree that results in a diagram that does not correspond to a tree, because it has an isolated node, but to a forest with two components.  Bottom: the same sandwich product showing the trees and forest involved.  }
\label{f:sandwich.non.tree}
\end{figure}

Likewise, the sandwich product allows the multiplication of two forests relative to a third. As with non-binary trees, there will be some products that are not defined, because the numbers of nodes in the diagrams do not match (if the numbers of non-leaf vertices in the forest are not equal).  But additionally, even if this is satisfied so that the diagram product is defined, the product may not result in a forest.  And this also applies to sandwich products of non-binary trees: the product of two non-binary trees, or forests, relative to a third, may result in a diagram that does not even represent a forest, in that it violates the condition in Theorem~\ref{t:forests.partitions}(iii).  An example is shown in Figure~\ref{f:sandwich.not.closed}.

\begin{figure}[ht]
\includegraphics[width=.7\textwidth]{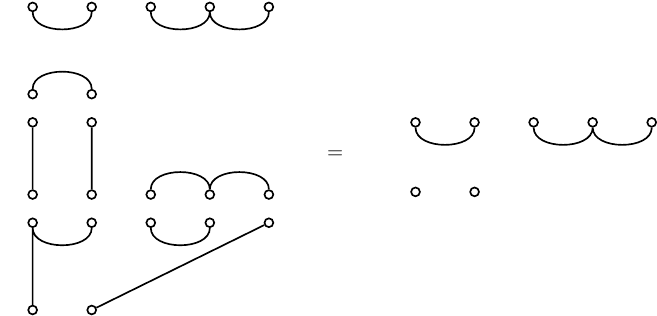}
\caption{A sandwich product involving diagrams from two non-binary trees relative to a third, and which results in a diagram that does not define a tree or a forest, because it violates the condition in Theorem~\ref{t:forests.partitions}(iii).  This product diagram has $m=2$, $\tau=3$ (one more than the number of singleton nodes), and $\ell(\pi)=4$, so does not satisfy $m=\ell(\pi)-\tau$.}
\label{f:sandwich.not.closed}
\end{figure}

\section{Discussion}\label{s:discussion}

The link between phylogenetic trees and algebraic structures such as Brauer monoids and partition monoids that we have described gives rise to a wealth of questions that need further exploration.  

To begin, some opportunities for further development were raised by Holmes and Diaconis in 1998~\cite{diaconis1998matchings}.  For instance, they suggest the representation of a tree as a matching that they introduced (and treating the matched pairs as 2-cycles) allows for the use of multiplication in the symmetric group to create a random walk in tree space.  They also suggested the use of multiplication of matchings through use of the Brauer algebra, that we have developed further in this paper in a direction they perhaps did not anticipate (by preserving the leaves along the top of a diagram to create an unbalanced but biologically interpretable model).

The developments here, representing the matching (or partition more generally) as a Brauer or partition diagram, allow a more targeted random walk to be defined that preserves certain key structures of the trees.  That is, a random step can be performed by acting on the top or bottom of a diagram by an element of $\sS_n$ or $\sS_m$, and allow movement along an $\RR$-class or $\LL$-class, which preserves certain structural properties of the tree, as described in Section~\ref{s:structure}.

There are many further questions that warrant exploration, some of which we list here:

\begin{enumerate}
	\item Questions about products of trees in the sandwich semigroup:
	\begin{enumerate}
		\item We have defined a way to multiply two binary trees relative to a third, via the sandwich product.  How do features of trees (such as their balance) behave when they are multiplied together?  
		\item Are some properties of trees `closed' under multiplication?  If two trees share the same property, when does their product share the same property?
		\item A feature of the sandwich semigroup construction is that it creates a product \emph{relative} to a fixed tree.  Is there some way to exploit this feature in phylogenetics? For instance, is it possible that gene trees may be constrained to be inside some neighbourhood defined by the product relative to the species tree (see \cite{szollosi2015inference} for discussion of the challenges understanding this relationship)?
		\item The sandwich semigroup product works naturally for binary trees, whose Brauer diagrams have predictable numbers of nodes.  Given the comments in Section~\ref{s:non-binary.partitions}\ref{s:forests} and Figure~\ref{f:sandwich.not.closed}, are there subclasses of forests (other than binary trees) for which the product is defined and closed? 
	\end{enumerate}
	\item Questions about the regular subsemigroup:
	\begin{enumerate}
		\item To what extent are features of trees preserved within the regular subsemigroup corresponding to the tree? Using balance as an example again, are the trees in the subsemigroup all close to being balanced under one of the standard balance measures?
		\item How do topological operations on a tree, such as the Nearest Neighbour Interchange (NNI), interact with the structures described here? For instance, is the NNI neighbourhood of a tree contained in the regular subsemigroup of the sandwich semigroup at $T$?  Or as mentioned at the end of Section~\ref{s:sandwich}, are trees of rank 0 (with maximal cherries) always in the regular subsemigroup? 
	\end{enumerate}	
	\item Other links to tree space:
	\begin{enumerate}
		\item As noted in the Introduction, Holmes and Diaconis raised the prospect of using the product of trees to randomly move around tree space or search for an optimal solution to a problem.  The properties of this random walk, or one extended to non-binary trees or forests, are unknown.
		\item If we restrict attention only to Brauer diagrams that are \emph{planar}, that is, for which there are no crossings of lines, then we obtain another closed algebraic structure.  That is, in the sandwich semigroup corresponding to a tree with planar diagram, the product of two planar diagrams will remain planar.  Any diagram can be made planar by relabelling the leaves, so would this semigroup describe some sort of canonical representatives of tree space?
		\item Can operations on trees,  such as edge-cutting that take a tree and produce a forest (but also even tree rearrangement operations such as NNI), be implemented by a sandwich product, or an action by an element of the partition monoid?  For instance, can sandwich products such as that seen in Figure~\ref{f:sandwich.non.tree} be controlled systematically?
		\item There is an interesting correspondence between ``augmented perfect matchings of $[2n-2]$ containing $\ell$ wiggly lines'' and phylogenetic trees on $n$ leaves and $n-\ell-1$ internal and root vertices, recently described in~\cite{price2020augmented}.   It would be interesting to investigate how such augmented matchings link to diagrams, and the set of partitions $\Lambda_{[n]}^{(\ge 2)}$ described here for such trees in Theorem~\ref{t:trees.partitions.correspond}. 
	\end{enumerate}
	\item Other links to semigroup theory and the Brauer algebra:
	\begin{enumerate}
		\item Within semigroup theory, and indeed broadly within other algebraic areas, the idempotents (elements $x$ for which $x^2=x$) play a very important role.  What are the idempotents within the sandwich semigroup, and what relationships do the corresponding trees share?  A formula for the number of idempotents in the sandwich semigroup is known~\cite[Theorem 6.18]{dolinka2021sandwich}: 
		for instance, for a 6 leaf tree with 3 cherries (rank 0), there are 45 idempotents in the corresponding sandwich semigroup (which is also the number of rank 0 trees on six leaves).  A characterisation could be illuminating.  Narrowing it down further, the mid-identities (elements that are idempotents when regular but satisfying additional conditions) could have a concrete phylogenetic relationship to the original tree.
		\item It would be interesting to study further the restricted Green's relations in which the action on left and right is from a subgroup or subsemigroup (like our $\sS_n$ action in Section~\ref{s:sandwich}).
		\item The general relations in the Brauer algebra (see Appendix~\ref{s:brauer.relations}) contain a central generator $\lambda$ that we ignore by setting it to be 1.  In the algebra, it tracks loops that can occasionally be generated by diagram concatenation, as occurs in the example in Figure~\ref{f:brauer.product.eg}. Is there phylogenetically relevant information that can be captured by loops arising in products, and that might benefit from use of the full $\Br_n(\lambda)$?
	\end{enumerate}
\end{enumerate}

There are several opportunities to extend the ideas in this paper in different directions.  
For instance, is it possible to represent phylogenetic \emph{networks} within a diagram semigroup framework?
And coming from the algebraic point of view, is there a role for other monoidal and categorical systems, such as those described in \cite{dolinka2021sandwich} (compare also Loday \cite{MR1958907}) to play within phylogenetics, or other scientific and combinatorial problems?  

Clearly, there is opportunity for exploration and development of this approach at an algebraic, combinatorial, phylogenetic, and computational level.  Diagram semigroups, monoids, algebras, and categories have found numerous diverse and powerful applications within mathematics and physics, and it is exciting to think that they may open new doors to phylogeneticists.

\subsection*{Acknowledgements}
When the authors began exploring the potential for modeling phylogenetic trees with Brauer diagrams, they made the fortunate decision to ask James East for advice.  They are tremendously grateful for his patient explanations of ideas such as the sandwich semigroup, and for pointing them to important references.

\appendix

\section{Presentation for the Brauer algebra}
\label{s:brauer.relations}

For each $n=1,2,\cdots$ the Brauer algebra ${\mathfrak B}_n(\lambda)$ is the complex algebra generated by the elements
$\{s_1,s_2,\cdots, s_{n\!-\!1}, e_1,e_2,\cdots, e_{n\!-\!1}\}$
subject to the defining relations~\cite{Brauer1937algebras}:
\begin{align*}
s_i^2= &\,1\,,\quad s_i s_{i\!+\!1}s_i = s_{i\!+\!1}s_is_{i\!+\!1}\,,\\
s_i s_j = & s_j s_i \,\quad \mbox{for} \quad |i-j|>1 \,;\\
e_i^2= &\,\lambda e_i\,,
\quad e_i e_{i\pm1}e_i = e_i\,, \\
\quad s_is_{i\pm1}e_i = &\, e_{i\pm1}e_i\,,\quad
e_is_{i\pm1}s_i = e_ie_{i\pm1}\,,\\
e_i e_j = &\, e_j e_i \,, \quad s_i s_j = s_j s_i\,,
\quad s_i e_j =  e_j s_i \,, \quad \mbox{for} \quad |i-j|>1 \,;\\
s_ie_i=&\,e_is_i = e_i\,, \quad e_i s_{i\pm 1} e_i = e_i\,.
\end{align*}
For the generalization of these relations arising from the connection with binary trees,  to the full partition algebra underlying the extension to non-binary trees and forests, 
see \cite{martin1996spa,HalversonRam2005pa}\,.

\end{document}